\documentclass[11pt,a4paper]{article}
\usepackage{amssymb}
\usepackage{amsmath}
\usepackage{amsfonts}
\usepackage{bbm}
\usepackage{amsthm}
\usepackage{mathrsfs}
\usepackage{hyperref}
\usepackage{marvosym}
\usepackage{color}
\usepackage[margin=2.5cm]{geometry}
\usepackage[all,cmtip]{xy}
\usepackage[utf8]{inputenc}
\usepackage{graphicx}

\usepackage{upgreek}
\usepackage{rotating}

\definecolor{darkred}{rgb}{0.8,0.1,0.1}
\hypersetup{
     colorlinks=false,         % false: boxed links; true: colored links,false is default
     linkcolor=darkred,
     citecolor=blue,
}

\theoremstyle{plain}
\newtheorem{theo}{Theorem}[section]
\newtheorem{lem}[theo]{Lemma}
\newtheorem{propo}[theo]{Proposition}
\newtheorem{cor}[theo]{Corollary}

\theoremstyle{definition}
\newtheorem{defi}[theo]{Definition}
\newtheorem{exx}[theo]{Example}
\newenvironment{ex}
  {\pushQED{\qed}\exx}
  {\popQED\endexx}

\newtheorem{remm}[theo]{Remark}
\newenvironment{rem}
  {\pushQED{\qed}\remm}
  {\popQED\endremm}

\numberwithin{equation}{section}

\def\nn{\nonumber}

\def\bbK{\mathbb{K}}
\def\bbR{\mathbb{R}}

\def\bbN{\mathbb{N}}
\def\bbZ{\mathbb{Z}}
\def\bbT{\mathbb{T}}

\def\id{\mathrm{id}}

\def\dd{\mathrm{d}}

\def\1{\mathbbm{1}}

\def\op{\mathrm{op}}

\def\pr{\mathrm{pr}}

\def\Loc{\mathsf{Loc}}

\def\Alg{\mathsf{Alg}}

\def\Str{\mathsf{Str}}
\def\dgVec{\mathsf{dgVec}}
\def\dgAlg{\mathsf{dgAlg}}
\def\DD{\mathsf{D}}

\def\ext{\mathrm{ext}}

\def\AA{\mathfrak{A}}

\def\BB{\mathfrak{B}}
\def\QQ{\mathfrak{Q}}

\def\Ranpi{\mathrm{Ran}_{\pi}^{}}
\def\hoRanpi{\mathrm{hoRan}_{\pi}^{}}
\def\holim{\mathrm{holim}}

\def\U{\mathrm{U}_{\pi}^{}}
\def\hoU{\mathrm{hoU}_{\pi}^{}}

\newcommand\mycom[2]{\genfrac{}{}{0pt}{}{#1}{#2}}

\def\sk{\vspace{2mm}}

\title{%
Quantum field theories on categories fibered in groupoids
}

\author{%
Marco Benini$^{1,a}$ and Alexander Schenkel$^{2,b}$ \vspace{3mm}\\
{\small $^1$ Institut f\"ur Mathematik, Universit\"at Potsdam,}\\
{\small Karl-Liebknecht-Str.~24-25, 14476 Potsdam, Germany.}\vspace{2mm}\\
{\small ${}^2$ School of Mathematical Sciences, University of Nottingham,}\\
{\small University Park, Nottingham NG7 2RD, United Kingdom.}\vspace{3mm}\\
{\footnotesize \texttt{Email:} 
$^a$\texttt{mbenini87@gmail.com}, $^b$\texttt{alexander.schenkel@nottingham.ac.uk}}
 }

\date{July 2017}

\begin{document}

\maketitle

\begin{abstract}
\noindent We introduce an abstract concept of quantum field theory on categories fibered in groupoids over the category of spacetimes. This provides us with a general and flexible framework to study quantum field theories defined on spacetimes with extra geometric structures such as bundles, connections and spin structures. Using right Kan extensions, we can assign to any such theory an ordinary quantum field theory defined on the category of spacetimes and we shall clarify under which conditions it satisfies the axioms of locally covariant quantum field theory. The same constructions can be performed in a homotopy theoretic framework by using homotopy right Kan extensions, which allows us to obtain first toy-models of homotopical quantum field theories resembling some aspects of gauge theories.
\end{abstract}


%{\baselineskip=12pt
%\setcounter{tocdepth}{1}
\tableofcontents
%}

%\bigskip

\newpage

%%%%%%%%%%%%%%%%%%%%%%%%%%%%%%%%%%%%%%%%%%%%%%%%
%%%%%%%%%%%%%%%%%%%%%%%%%%%%%%%%%%%%%%%%%%%%%%%%

\section{\label{sec:intro}Introduction and summary}
A locally covariant quantum field theory in the original sense of \cite{Brunetti}
is a functor $\AA : \Loc\to \Alg$ which assigns algebras of quantum observables 
to globally hyperbolic Lorentzian manifolds (i.e.\ spacetimes) 
subject to a collection of physically motivated axioms.
In practice, however, it is often convenient to slightly generalize
this framework and consider functors $\AA : \Str\to \Alg$ defined
on a category $\Str$ of spacetimes with additional geometric structures.
For example, Dirac quantum fields are typically defined on the category of globally hyperbolic
Lorentzian spin manifolds (cf.\ \cite{Verch,DHP, Sanders}) and charged quantum fields in the presence of
background gauge fields are defined on a category of principal bundles with connections over spacetimes
(cf.\ \cite{Zahn, SchenkelZahn}). 
A common feature of these and similar examples appearing throughout the literature
is that there exists a projection functor $\pi:\Str\to\Loc$ from structured spacetimes to spacetimes
which forgets the extra geometric structures. In examples, the functor $\pi:\Str\to\Loc$ 
exhibits special properties in the sense that 1.)~geometric structures defined over a spacetime 
$M^\prime$ admit pullbacks along $\Loc$-morphisms $f:M\to M^\prime$ and 2.)~the fibers $\pi^{-1}(M)$
of geometric structures over a spacetime $M$ are groupoids. (The morphisms of these
groupoids should be interpreted as gauge transformations between geometric structures over $M$.) 
In technical terms, this means that $\pi:\Str\to\Loc$ is a category fibered in groupoids.
\sk

In this paper we abstract these examples and study quantum field theories $\AA :\Str\to\Alg$
defined on categories fibered in groupoids $\pi:\Str\to \Loc$ from a model-independent perspective.
We shall show that to any such theory one can assign (via a universal construction 
called right Kan extension) a functor $\U\AA :\Loc\to\Alg$ defined on the category of spacetimes $\Loc$.
The relationship between $\U\AA$ and $\AA$
is specified by a (universal) diagram of functors
\begin{flalign}\label{eqn:diagramintro}
\xymatrix{
\ar[dr]_-{\pi}\Str \ar[rr]^-{\AA} &\ar@{<=}[d]_-{\epsilon}& \Alg\\
&\Loc\ar[ru]_-{\U\AA}&
}
\end{flalign}
which commutes up to a natural transformation $\epsilon$ that embeds $\U\AA \circ \pi$ 
as a subtheory of $\AA$. 
We will show that the right Kan extension $\U\AA : \Loc\to \Alg$ assigns interesting 
algebras $\U\AA(M)$ to spacetimes $M$, 
which one may interpret as gauge invariant combinations
of classical observables for the geometric structures over $M$ and quantum observables
of the original theory $\AA:\Str\to \Alg$ corresponding to all possible structures over $M$.
In physical terminology, this means that the {\em background} geometric structures
of the theory $\AA : \Str\to \Alg$ on $\Str$ are promoted via the right Kan extension
to (classical) {\em degrees of freedom} of the theory $\U\AA:\Loc\to\Alg$ on $\Loc$.
The latter perspective has the advantage that the 
extra geometric structures do not have to be chosen 
a priori for assigning an algebra via $\AA :\Str\to\Alg$, 
but they may be selected later by a suitable choice of state on the algebra $\U\AA(M)$.
We will prove a theorem providing sufficient (and in some cases also necessary) conditions
on the category fibered in groupoids $\pi:\Str\to\Loc$ such that the right Kan extension 
$\U\AA:\Loc\to\Alg$ satisfies the axioms of locally covariant quantum field theory \cite{Brunetti}.
In particular, we find that the isotony axiom is often violated by the
typical examples $\pi:\Str\to\Loc$ of spacetimes equipped with additional geometric structures  
considered in the literature. 
Such feature is similar to the isotony violations observed in models of quantum gauge theories 
\cite{DL,SDH,BDS,BDHS,BSS,BBSS,Benini}.
\sk

In very special instances, our general construction reduces to the assignment of the fixed-point theory
of a locally covariant quantum field theory with ``global gauge group'', which has 
been studied by Fewster in \cite{FewsterAuto}. A quantum field theory $\BB :\Loc\to\Alg$
on $\Loc$ together with a representation $\eta :G\to \mathrm{Aut}(\BB)$ of a group $G$ in 
terms of automorphisms is equivalent to a quantum field theory $\AA : \Loc\times G\to \Alg$
on the trivial category fibered in groupoids $\pi: \Loc\times G\to\Loc$, where all fibers are given
by $G$ (regarded as a groupoid with only one object). The right Kan extension
$\U\AA : \Loc\to\Alg$ of such a theory is then precisely the fixed-point theory of $(\BB,\eta)$.
From this perspective, our construction may  also be interpreted as a generalization
of the assignment of fixed-point theories to situations where the ``global gauge group'' 
$G$ is replaced by ``local gauge groupoids'' $\pi^{-1}(M)$, i.e.\ a family of
groupoids depending on the underlying spacetime $M$. A very intriguing (and, most likely, challenging)
task would be to understand which parts of the program of Doplicher, Haag and Roberts \cite{DHR1,DHR2}
can be generalized to our situation. For an extension of these techniques to Lorentzian manifolds see \cite{Ruzzi,BR}.
In a first attempt, this could be simplified by studying one of the following two distinct scenarios:
a)~The symmetries are described by a ``global gauge groupoid'', i.e.\ quantum field theories on the trivial category fibered in 
groupoids $\pi: \Loc\times \mathcal{G}\to\Loc$, where $\mathcal{G}$  is now any groupoid.
b)~The symmetries are described by a ``local gauge group'', i.e.\ quantum field theories on a 
category fibered in  groupoids $\pi: \Str \to\Loc$, where all fibers $\pi^{-1}(M)$ are ($M$-dependent) 
groupoids with only one object. Such questions lie beyond the scope of this 
paper because, in contrast to our purely algebraic and categorical approach, they presumably 
also require a careful treatment of functional analytical aspects. 
\sk

In the second part of this paper we go beyond the standard framework of locally 
covariant quantum field theory by adding a homotopy theoretical flavor to
our constructions. This generalization is motivated by the well-known mathematical 
fact that the local-to-global behavior of gauge theories, which is captured by
the concept of {\em descent} for stacks, see e.g.\ \cite{Hollander}, is necessarily
of a homotopical (or higher categorical) nature.
In \cite{BSShomotopy}, we initiated the development of a 
homotopical generalization of locally covariant quantum field theory, 
where observable algebras are replaced by higher algebraic structures such as 
differential graded algebras or cosimplicial algebras. Using simple toy-models
given by classical (i.e.\ not quantized) and non-dynamical Abelian gauge theories,
we confirmed that this approach allows for a homotopically refined
version of Fredenhagen's ``universal algebra'' construction \cite{Fre1,Fre2,Fre3}
which is suitable for gauge theories.
What was missing in \cite{BSShomotopy} is a study 
of the crucial question how the axioms of locally covariant quantum field theory
may be implemented in a homotopically meaningful way.
In this work we provide a first answer to this question by
constructing toy-models of homotopical quantum field theories
via a homotopical generalization of the right Kan extension.
Concretely, we consider a category fibered in groupoids $\pi : \Str \to \Loc$
and a non-homotopical quantum field theory $\AA : \Str\to \Alg$ on it.
As an example, one may think of $\AA$ as a charged matter quantum field theory 
coupled to background gauge fields that are encoded in the groupoids $\pi^{-1}(M)$.
Regarding $\AA:\Str\to\dgAlg$ as a trivial homotopical quantum field theory via 
the embedding $\Alg \to \dgAlg$ of algebras into differential graded algebras 
concentrated in degree $0$, we produce a (generically) non-trivial homotopical 
quantum field theory $\hoU\AA:\Loc \to \dgAlg$ on $\Loc$ via 
the homotopy right Kan extension. We observe that this theory assigns
to a spacetime $M$ the differential graded algebra 
$\hoU\AA(M) = C^\bullet(\pi^{-1}(M);\AA)$ underlying the groupoid cohomology 
of $\pi^{-1}(M)$ with values in the functor $\AA:\Str\to\Alg$.
Its zeroth cohomology is precisely the algebra $\U\AA(M)$
assigned by the ordinary right Kan extension (i.e.\ an algebra of gauge invariant 
combinations of classical gauge field and quantum matter field observables)
and its higher cohomologies encode more detailed aspects of the action of the ``gauge groupoids''
$\pi^{-1}(M)$ on $\AA$. Notice that the higher cohomologies are not visible in a 
non-homotopical approach and hence are novel features arising within it. Unfortunately,
a satisfactory interpretation of the physics encoded in such higher cohomologies is not entirely developed yet.
\sk

Using our toy-models $\hoU\AA:\Loc \to \dgAlg$ for homotopical quantum field theories
we investigate to which extent they satisfy the axioms of locally covariant quantum field theory
\cite{Brunetti}. Our first observation is that $\hoU\AA:\Loc \to \dgAlg$ is in general not a strict functor.
It is only a functor `up to homotopy' in the sense that for two composable $\Loc$-morphisms
$f$ and $f^\prime$ there exists a cochain homotopy 
$\hoU\AA(f^\prime) \circ \hoU\AA(f) \sim \hoU\AA(f^\prime\circ f)$ controlling compositions.
Similarly, we observe that the causality axiom and (under suitable conditions) the time-slice axiom
hold `up to homotopy'. Such weaker notions of the axioms of locally covariant quantum field
theory are mathematically expected because the strict axioms are unstable under weak equivalences.
Concretely, once we are given a model for a homotopical quantum field theory satisfying the
strict axioms, we could pass to a weakly equivalent description that will only satisfy the axioms `up to homotopy'.
We would like to emphasize that the concept of causality `up to homotopy' is different from
the weakened causality condition for interlinked regions discovered in \cite{Buc1,Buc2}.
In fact, the former descends to strict causality on the level of cohomologies and in particular
on the level of gauge invariant observables (i.e.\ the zeroth cohomology).
\sk

Our studies also indicate that there seems to exist a refinement
of the `up to homotopy' axioms by higher homotopies and coherence conditions.
This means that one may {\em choose} particular cochain homotopies which enforce the
`up to homotopy' axioms and control their iterations (e.g.\ multiple compositions of morphisms
or multiple commutations of spacelike separated observables) by higher cochain homotopies and coherences.
From a homotopical perspective, it is natural to add all those (higher) homotopies and their coherences
to the data defining a homotopical quantum field theory. It is however very hard to deal
with such structures by using only elementary categorical techniques. 
To cope with (higher) homotopies and their coherences systematically, 
one needs the machinery of colored operads (see e.g.\ \cite{Operad}). 
Hence, our results point towards the usefulness of colored operads in the formulation
of locally covariant quantum field theory and its {\em coherent} homotopical 
generalization. This operadic perspective will be developed in our future works.
It is worth to emphasize the differences between the non-coherent approach to 
homotopical quantum field theory employed in the present paper and the aforementioned coherent one: 
The former allows us to assign interesting differential graded algebras to spacetimes,
whose cohomologies capture, in addition to gauge invariant observables, 
further information about the action of gauge transformations. Moreover, all information
encoded in these cohomologies satisfies the locally covariant quantum field theory axioms strictly. 
This is already quite satisfactory if one is mainly interested in cohomological information 
(even more so, when gauge invariant observables, i.e.\ the zeroth cohomology, are the main object of concern). 
On the contrary, in order to perform certain constructions, 
it becomes crucial to keep track of all (higher) homotopies and their coherences 
and therefore a {\em coherent} homotopical generalization (in the sense explained above) 
of locally covariant quantum field theory becomes necessary. 
For example, this is the case when one is confronted with
questions related to local-to-global properties, 
e.g.\ generalizations of Fredenhagen's ``universal algebra'' construction.
The reason is that such constructions involve colimits over {\em commutative} diagrams 
associated to embeddings of spacetime regions, whose homotopical generalization must be in terms of 
{\em homotopy coherent commutative} diagrams. Notice that this is very similar to the formulation of
descent for stacks in non-strict models, e.g.\ in terms of pseudo-functors (cf.\ \cite{Vistoli}).
\sk

The outline of the remainder of this paper is as follows:
In Section \ref{sec:setup} we review some basic aspects of
categories fibered in groupoids $\pi:\Str\to\Loc$ over the spacetime category $\Loc$
and introduce a notion of quantum field theory on them. We will also show that
many examples of quantum field theories defined on spacetimes
with extra geometric structures appearing throughout the literature fit into our framework.
In Section \ref{sec:Kan} we compute the right Kan extension of a quantum field theory
$\AA: \Str\to\Alg$ on structured spacetimes along the projection functor $\pi:\Str\to\Loc$
and thereby obtain candidates $\U\AA : \Loc\to \Alg$ for quantum field theories on $\Loc$.
In Section \ref{sec:properties} we prove a theorem providing sufficient (and in some cases also necessary) 
conditions on the category fibered in groupoids $\pi:\Str\to\Loc$
such that the right Kan extension $\U\AA : \Loc\to \Alg$ satisfies the axioms of 
locally covariant quantum field theory.
We will confirm by examples that there exist right Kan extensions which satisfy the
causality and time-slice axioms, while the isotony axiom is typically violated.
A homotopical generalization of these constructions is studied in Section \ref{sec:homotopyKan}
and its properties are studies in Section  \ref{sec:homproperties}. As a result, we construct
first toy-models of homotopical quantum field theories via homotopy right Kan extensions. 
Appendix \ref{app:dgAlg} contains some standard material on the homotopy theory of 
differential graded vector spaces and differential graded algebras, which is used in the main text.

\section{\label{sec:setup}Setup}
Let us denote by $\Loc$ the category of $m$-dimensional oriented, 
time-oriented and globally hyperbolic Lorentzian manifolds
with morphisms given by orientation and time-orientation
preserving causal, isometric and open embeddings.
Physically, $\Loc$ describes the category of spacetimes without
additional geometric structures such as bundles (with connections)
or spin structures. In order to allow for such additional geometric structures,
we consider a category $\Str$, which describes the structures of interest as well as their symmetries,
together with a functor $\pi : \Str \to \Loc$ that assigns the underlying spacetime.
\sk

A quantum field theory on structured spacetimes is then given by a functor
\begin{flalign}
\AA : \Str \longrightarrow \Alg
\end{flalign}
to the category of unital associative algebras over a (fixed) 
field $\bbK$ with morphisms given by unital algebra homomorphisms. 
We shall assume the standard axioms of locally covariant quantum field theory \cite{Brunetti},
adapted to the category $\Str$ and the functor $\pi : \Str\to\Loc$. 
\begin{defi}\label{def:QFT}
A functor $\AA: \Str \to \Alg$ is called a {\em quantum field theory} on $\pi : \Str\to \Loc$ if
the following axioms are fulfilled.
\begin{itemize}
\item {\em Isotony:} For every $\Str$-morphism $g : S \to S^\prime$, the $\Alg$-morphism
$\AA(g) : \AA(S)\to \AA(S^\prime)$ is a monomorphism.

\item {\em Causality:} Let $S_1 \stackrel{g_1}{\longrightarrow} S \stackrel{g_2}{\longleftarrow} S_2 $ 
be a $\Str$-diagram, such that its projection via $\pi$ to $\Loc$ 
$\pi(S_1) \stackrel{\pi(g_1)}{\longrightarrow} \pi(S) \stackrel{\pi(g_2)}{\longleftarrow} \pi(S_2) $
is causally disjoint, i.e.\ the images of $\pi(g_1)$ and $\pi(g_2)$ are causally disjoint subsets of $\pi(S)$.
Then the induced commutator
\begin{flalign}
[\,\cdot\,,\,\cdot\,] \circ \big(\AA(g_1)\otimes \AA(g_2)\big) \,:\, \AA(S_1)\otimes \AA(S_2) \longrightarrow \AA(S)
\end{flalign}
is zero.

\item {\em Time-slice:} Let $g : S \to S^\prime$ be a $\Str$-morphism, such that
its projection $\pi(g) : \pi(S) \to \pi(S^\prime)$ via $\pi$ to $\Loc$ is a Cauchy $\Loc$-morphism, 
i.e.\ the image of $\pi(g)$ contains a Cauchy surface of $\pi(S^\prime)$. Then
the $\Alg$-morphism $\AA(g) : \AA(S)\to \AA(S^\prime)$ is an isomorphism.
\end{itemize}
Moreover, we shall always assume that, for each object $S$ in $\Str$,
$\AA(S)$ is not a terminal object in $\Alg$. Equivalently, this means that 
the unit element $\1\in \AA(S)$ is different from the zero element $0\in\AA(S)$, i.e.\ $\1\neq 0$,
for each object $S$ in $\Str$.
\end{defi}
\begin{rem}
Notice that a quantum field theory $\AA:\Loc\to\Alg$ on the identity functor
$\id_{\Loc}:\Loc \to\Loc$ is  a locally covariant quantum field theory 
in the sense of \cite{Brunetti}.
\end{rem}

The case where we just assume any functor $\pi : \Str \to \Loc$ will turn out to be too generic
to allow for interesting model-independent constructions. In many examples of interest, 
some of which we shall review below, it turns out that any object $S^\prime$ in $\Str$
may be pulled back along a $\Loc$-morphism $f : M\to \pi(S^\prime)$, giving rise to
an object $f^\ast S^\prime$ in $\Str$ with $\pi(f^\ast S^\prime) = M$
and a $\Str$-morphism $f_\ast : f^\ast S^\prime \to S^\prime$ such that
$\pi(f_\ast) =f : M\to \pi(S^\prime)$. Existence of pullbacks can be formalized in 
terms of fibered categories, see e.g.\ \cite[Section 3]{Vistoli} for an introduction.
Let us briefly review the main definitions relevant for our work.
\begin{defi}\label{def:cartesian}
A $\Str$-morphism $g : S\to S^\prime$ is called {\em cartesian} if
for any $\Str$-morphism $g^\prime : \widetilde{S} \to S^\prime$
and any $\Loc$-morphism $f : \pi(\widetilde{S}) \to \pi(S)$,
such that the $\Loc$-diagram
\begin{flalign}
\xymatrix{
\ar[dr]_-{f}\pi(\widetilde{S}) \ar[rr]^-{\pi(g^\prime)} && \pi(S^\prime)\\
&\pi(S) \ar[ru]_-{\pi(g)}&
}
\end{flalign}
commutes, there exists a unique $\Str$-morphism $\widetilde{g} : \widetilde{S} \to S$,
such that $\pi(\widetilde{g}) =f$ and the $\Str$-diagram
\begin{flalign}
\xymatrix{
\ar@{-->}[dr]_-{\exists !\, \widetilde{g}} \widetilde{S} \ar[rr]^-{g^\prime} && S^\prime\\
&S \ar[ru]_-{g}&
}
\end{flalign}
commutes.
If $g : S\to S^\prime$ is a cartesian $\Str$-morphism, we also say that $S$ is a {\em pullback}
of $S^\prime$ to $\pi(S)$ (along the $\Loc$-morphism $\pi(g) : \pi(S)\to \pi(S^\prime)$).
\end{defi}
\begin{rem}
As a direct consequence of the universal definition of cartesian $\Str$-morphisms,
it follows that any two pullbacks of $S^\prime$ to $M$ along a $\Loc$-morphism
$f:M\to\pi(S^\prime)$  (if they exist) are isomorphic via a unique isomorphism.
\end{rem}
\begin{defi}\label{def:fiberedcategory}
A functor $\pi : \Str\to\Loc$ is called a {\em fibered category} over $\Loc$ if for any
$\Loc$-morphism $f : M\to M^\prime$ and any object $S^\prime$ in $\Str$ with $\pi(S^\prime)=M^\prime$
there exists a cartesian $\Str$-morphism $g:S\to S^\prime$ such that
$\pi(g) =f: M\to M^\prime$. 
\end{defi}
\begin{defi}\label{def:categoryfiberedingrpd}
A functor $\pi : \Str\to\Loc$ is called a {\em category fibered in groupoids} over $\Loc$
if it is a fibered category over $\Loc$ and additionally $\pi^{-1}(M)$ is a groupoid, for each object $M$ in $\Loc$.
Here $\pi^{-1}(M)$ is the subcategory of $\Str$ with objects given by all objects $S$ in $\Str$ such that
$\pi(S) =M$ and morphisms given by all $\Str$-morphisms $g : S\to S^\prime$ such that
$\pi(g) = \id_M : M\to M$.
\end{defi}

We finish this section by providing some examples of categories fibered in groupoids over $\Loc$,
which were used in the literature to describe quantum field theories that are
defined on spacetimes with extra geometric structures or admit some additional symmetries.

\begin{ex}[Spin structures]\label{ex:spin}
Assume that the spacetime dimension is $m\geq 4$.
Let $\mathsf{SLoc}$ be the category of $m$-dimensional oriented, 
time-oriented and globally hyperbolic Lorentzian spin manifolds.
Its objects are all tuples $(M,P,\psi)$, where $M$ is an object in $\Loc$,
$P$ is a principal $\mathrm{Spin}_0(1,m-1)$-bundle over $M$
and $\psi : P \to FM$ is a $\mathrm{Spin}_0(1,m-1)$-equivariant
bundle map (over $\id_M$) to the pseudo-orthonormal oriented 
and time-oriented frame bundle $FM$ over $M$.
(The right $\mathrm{Spin}_0(1,m-1)$-action on $FM$ is induced by the double covering
group homomorphism $\rho: \mathrm{Spin}_0(1,m-1) \to \mathrm{SO}_0(1,m-1)$.)
A morphism $g : (M,P,\psi) \to (M^\prime,P^\prime,\psi^\prime)$ in $\mathsf{SLoc}$
is a principal $\mathrm{Spin}_0(1,m-1)$-bundle morphism $g : P\to P^\prime$ covering
a $\Loc$-morphism $f: M\to M^\prime$, such that $ \psi^\prime \circ g = f_\ast \circ \psi$, where
$f_\ast : FM\to FM^\prime$ is the pseudo-orthonormal oriented and time-oriented
frame bundle morphism induced by the $\Loc$-morphism $f: M\to M^\prime$.
\sk

There is an obvious functor $\pi : \mathsf{SLoc}\to\Loc$ which forgets the spin structure, i.e.\
$\pi(M,P,\psi) = M$ and $\pi(g) = f$. The fiber $\pi^{-1}(M)$ over any object $M$ in $\Loc$
is a groupoid, because principal bundle morphisms covering the identity are isomorphism.
Moreover, $\pi : \mathsf{SLoc}\to\Loc$ is a fibered category and thus a category fibered in groupoids: 
Given any object $(M^\prime,P^\prime,\psi^\prime)$ in $\mathsf{SLoc}$ and any $\Loc$-morphism
$f : M\to M^\prime$, we pull $P^\prime$ back to a principal $\mathrm{Spin}_0(1,m-1)$-bundle
$f^\ast P^\prime$ over $M$ and $FM^\prime$ to a principal $\mathrm{SO}_0(1,m-1)$-bundle
$f^\ast FM^\prime$ over $M$, where, as a consequence of the properties of $\Loc$-morphisms, 
the latter is isomorphic to the pseudo-orthonormal oriented 
and time-oriented frame bundle $FM$ over $M$. Composing the induced map
$f^\ast\psi^\prime : f^\ast P^\prime\to f^\ast FM^\prime$ with the isomorphism $f^\ast FM^\prime\simeq FM$
we obtain a $\mathrm{Spin}_0(1,m-1)$-equivariant bundle map $\psi : f^\ast P^\prime\to FM$, hence
an object $(M,f^\ast P^\prime,\psi)$ in $\mathsf{SLoc}$. The canonical principal $\mathrm{Spin}_0(1,m-1)$-bundle 
morphism $f_\ast : f^\ast P^\prime \to P^\prime$ covering the $\Loc$-morphism $f:M\to M^\prime$
defines a $\mathsf{SLoc}$-morphism $f_\ast : (M,f^\ast P^\prime,\psi)\to (M^\prime,P^\prime,\psi^\prime)$
and it is straightforward to verify that the latter is cartesian.
\sk

Examples of quantum field theories defined on $\pi : \mathsf{SLoc}\to \Loc$
include Dirac quantum fields, see e.g.\ \cite{Verch,DHP, Sanders}.
(In order to include also fermionic quantum field theories, Definition \ref{def:QFT} 
has to be generalized in the usual way to $\bbZ_2$-graded algebras, see e.g.\ \cite{BaerGinoux}.)
\end{ex}

\begin{ex}[Principal bundles (with connections)]\label{ex:gauge}
Fix any Lie group $G$. Let $\mathsf{B}G\Loc$ be the category with 
objects given by all pairs $(M,P)$, where $M$ is an object in $\Loc$ and $P$ is a principal $G$-bundle over $M$,
and morphisms $g : (M,P)\to (M^\prime,P^\prime)$ given by all principal $G$-bundle morphisms
$g : P\to P^\prime$ covering a $\Loc$-morphism $f:M\to M^\prime$.
There is an obvious functor $\pi : \mathsf{B}G\Loc \to\Loc$ which forgets the bundle data, 
i.e.\ $\pi(M,P) = M$ and $\pi(g) = f$. Using ordinary pullbacks of principal bundles as in Example \ref{ex:spin}, 
it is easy to show that $\pi : \mathsf{B}G\Loc \to\Loc$ is a category fibered in groupoids.
\sk

Let $\mathsf{B}G^{\mathrm{con}}\Loc$ be the category with 
objects given by all tuples $(M,P,A)$, where $M$ is an object in $\Loc$, $P$ is a principal $G$-bundle over $M$
and $A$ is a connection on $P$, and morphisms $g : (M,P,A)\to (M^\prime,P^\prime,A^\prime)$ 
given by all principal $G$-bundle morphisms $g : P\to P^\prime$ covering a $\Loc$-morphism $f:M\to M^\prime$
and preserving the connections, i.e.\ $g^\ast A^\prime = A$.
There is an obvious functor $\pi : \mathsf{B}G^{\mathrm{con}}\Loc \to\Loc$ which forgets the bundle and connection 
data,  i.e.\ $\pi(M,P,A) = M$ and $\pi(g) = f$. Using again ordinary pullbacks of principal bundles and also pullbacks of
connections, it is easy to show that $\pi : \mathsf{B}G^{\mathrm{con}}\Loc \to\Loc$ is a category fibered in groupoids.
\sk

Examples of quantum field theories defined on $\pi : \mathsf{B}G\Loc \to\Loc$
include dynamical quantum gauge theories on fixed but arbitrary principal bundles, see e.g.\ \cite{BDHS,BDS}.
Examples of quantum field theories defined on $\pi : \mathsf{B}G^{\mathrm{con}}\Loc \to\Loc$
include charged matter quantum field theories on fixed but arbitrary background gauge fields,
see e.g.\ \cite{Zahn,SchenkelZahn}.
\end{ex}

\begin{ex}[Global coframes]\label{ex:frame}
In \cite{Fewster1,Fewster2}, Fewster introduced the category $\mathsf{FLoc}$
of (co)framed spacetimes for studying model-independent aspects of the 
spin-statistics theorem. Objects in $\mathsf{FLoc}$ are all pairs $(M,e)$, where
$M$ is an $m$-dimensional manifold and $e = \{e^a \in\Omega^1(M) : a=0,1,\dots,m-1\}$
is a global coframe, such that the tuple $\pi(M,e) := (M, \eta_{ab}\,e^a\otimes e^b, e^0, e^0\wedge\dots\wedge e^{m-1})$
is an object in $\Loc$. Here $\eta_{ab} = \mathrm{diag}(1,-1,\dots,-1)_{ab}$ denotes the Minkowski metric.
A morphism $g : (M,e)\to (M^\prime,e^\prime)$ in $\mathsf{FLoc}$ is a smooth map
$g : M\to M^\prime$, such that $g^\ast {e^\prime}^a = e^a$, for all $a$, and
$\pi(g) := g : \pi(M,e)\to \pi(M^\prime,e^\prime)$ is a $\Loc$-morphism.
We obtain a functor $\pi: \mathsf{FLoc}\to \Loc$, which is easily seen to
be a category fibered in groupoids by pulling back global coframes. 
\end{ex}

\begin{ex}[Source terms]\label{ex:source}
For studying inhomogeneous Klein-Gordon quantum field theories in the presence of source terms
$J\in C^\infty(M)$, \cite{FewsterSchenkel} introduced the category $\mathsf{LocSrc}$.
Objects in $\mathsf{LocSrc}$ are all pairs $(M,J)$, where $M$ is an object in $\Loc$ and $J\in C^\infty(M)$,
and morphisms $g: (M,J)\to (M^\prime,J^\prime)$ are given by all
$\Loc$-morphisms $g: M\to M^\prime$ such that $g^\ast J^\prime = J$.
There is an obvious functor $\pi :\mathsf{LocSrc}\to\Loc$ which forgets the source terms,
i.e.\ $\pi(M,J) = M$ and $\pi(g) = g$. It is easy to check that $\pi :\mathsf{LocSrc}\to\Loc$
is a category fibered in groupoids.
\end{ex}

\begin{ex}[Global gauge transformations]\label{ex:globalgauge}
Let $G$ be a group. Interpreting $G$ as a groupoid with only one object 
(the automorphisms of this object are given by the elements $g\in G$ of the group),
we may form the product category $\Loc\times G$. Its objects
are the same as the objects in $\Loc$ and its morphisms
are pairs $(f,g) : M\to M^\prime$, where $f:M\to M^\prime$ is a $\Loc$-morphism
and $g\in G$ is a group element. Composition of morphisms is given by
$(f^\prime,g^\prime)\circ (f,g) = (f^\prime\circ f,g^\prime\,g)$ with $g^\prime\,g$ 
defined by the group operation on $G$, and the identity morphisms are $(\id_M,e)$ 
with $e\in G$ the identity element.
There is an obvious functor $\pi : \Loc\times G \to \Loc$ projecting onto $\Loc$,
i.e.\ $\pi(M) = M$ and $\pi(f,g) = f$. The fiber $\pi^{-1}(M)$ over any object $M$
in $\Loc$ is isomorphic to the groupoid $G$. It is easy to show that $\pi : \Loc\times G \to \Loc$
is a category fibered in groupoids. (Notice that every $\Loc\times G$-morphism is cartesian.)
\sk

Quantum field theories $\AA : \Loc\times G\to\Alg$ on $\pi : \Loc\times G \to \Loc$ are 
in one-to-one correspondence with ordinary quantum field theories $\BB:\Loc\to\Alg$ 
on $\Loc$ together with a representation $\eta : G \to \mathrm{Aut}(\BB)$ of the group $G$
in terms of automorphisms of $\BB$. 
(The concept of automorphism groups of locally covariant quantum field theories
was introduced and studied by Fewster in \cite{FewsterAuto}.)
Explicitly, given $\AA : \Loc\times G\to\Alg$, we define
a functor $\BB : \Loc\to \Alg$ by setting $\BB(M) :=\AA(M)$, for all objects $M$ in $\Loc$,
and $\BB(f) := \AA(f,e)$, for all $\Loc$-morphisms $f : M\to M^\prime$.
For $g\in G$, the natural isomorphism $\eta(g) : \BB\Rightarrow \BB$
is specified by the components $\eta(g)_M := \AA(\id_M,g)$, for all objects $M$ in $\Loc$.
It is easy to check naturality of these components and also that $\eta$ defines a representation of $G$.
Conversely, given $\BB:\Loc\to \Alg$ and $\eta: G\to \mathrm{Aut}(\BB)$,
we define a functor $\AA : \Loc\times G\to\Alg$ by setting $\AA(M):= \BB(M)$,
for all objects $M$ in $\Loc\times G$, and $\AA(f,g) := \BB(f) \circ \eta(g)_M$, for
all $\Loc\times G$-morphisms $(f,g):M\to M^\prime$. 
Since $\eta$ is a representation in terms of automorphisms of the functor $\BB$, 
it follows that $\AA$ is indeed a functor.
\sk

The automorphism group $\mathrm{Aut}(\BB)$ of a quantum field theory $\BB:\Loc\to \Alg$ on $\Loc$
was interpreted in \cite{FewsterAuto} as the ``global gauge group'' of the theory. See also \cite{Ruzzi,BR}
for a different point of view. Our fibered category approach thus includes also scenarios where one
is interested in (subgroups of) the ``global gauge group'' of ordinary quantum field theories and their actions. 
It is important to emphasize that the corresponding category fibered in groupoids
$\pi : \Loc\times G\to \Loc$ is extremely special when compared to our other examples above: 
1.)~Each fiber $\pi^{-1}(M)$ is isomorphic to $G$, i.e.\ a groupoid with only one object.
In our informal language from above, this means that there are no additional geometric 
structures attached to spacetimes but only additional automorphisms.
2.)~The fibers $\pi^{-1}(M)$ are the same for all spacetimes $M$. 
This justifies employing the terminology ``global gauge group'' for the present scenario.
\end{ex}
\begin{rem}\label{rem:KK}
Example \ref{ex:globalgauge} also captures a variant of Kaluza-Klein theories. (We are grateful
to one of the referees for asking us to address this point.)
Let us fix a compact oriented $k$-dimensional Riemannian manifold $K$, which we interpret as
the ``internal space'' of a Kaluza-Klein theory. The category $\Loc_{m+k}$ is defined
analogously to $\Loc$ by replacing $m$-dimensional manifolds with $m+k$-dimensional ones.
Let $\Loc_K$ be the subcategory of $\Loc_{m+k}$ whose
objects are of the form $M\times K$, with $M$ an object in $\Loc$ (i.e.\ $M$ is $m$-dimensional)
and $K$ our ``internal space'', and whose morphisms are of the form $(f,g) : M\times K\to M^\prime\times K$,
with $f:M\to M^\prime$ a $\Loc$-morphism. It follows that $g: K\to K$ is an orientation preserving
isometry of $K$ (in particular, note that $g$ is a diffeomorphism 
because it is an open embedding with compact image, hence $g(K)=K$), 
i.e.\ $g\in \mathrm{Iso}^+(K)$ is an element of the orientation preserving isometry group.
There is an obvious functor $\pi:\Loc_K\to \Loc$ projecting onto $\Loc$, i.e.\
$\pi(M\times K) = M$ and $\pi(f,g) =f$. Notice that $\pi:\Loc_K\to \Loc$ is a category fibered in groupoids
and as such isomorphic to $\pi: \Loc\times\mathrm{Iso}^+(K)\to\Loc$, cf.\ Example \ref{ex:globalgauge}.
Given any $m+k$-dimensional quantum field theory $\AA : \Loc_{m+k}\to \Alg$,
we may restrict it to the subcategory $\Loc_K$ and obtain a quantum field theory
on $\pi: \Loc_K\to\Loc$. This restriction may be interpreted as the first step
of a Kaluza-Klein construction because one introduces a fixed ``internal space'' $K$ 
and considers the theory only on spacetimes of the form $M\times K$, 
where $M$ is $m$-dimensional. A variant of Kaluza-Klein reduction is then given
by assigning the fixed-point theory of the $ \mathrm{Iso}^+(K)$-action. This
is captured by our general construction presented in the next section, see
Remark \ref{rem:physicalinterpretation}.
\end{rem}

\section{\label{sec:Kan}Kan extension}
Let $\pi : \Str \to \Loc$ be a category fibered in groupoids over $\Loc$
and let $\AA : \Str \to \Alg$ be a functor. 
In practice, $\AA$ will satisfy the quantum field theory axioms of Definition \ref{def:QFT},
but these are not needed for the present section.
The goal of this section is to canonically induce 
from this data a functor $\Ranpi \AA : \Loc \to \Alg$  
on the category $\Loc$, i.e.\ a candidate for a quantum field theory 
defined on spacetimes without additional structures.
Technically, our construction is a right Kan extension \cite[Chapter X]{MacLane}.
\begin{defi}
A {\em right Kan extension} of $\AA : \Str \to \Alg$ along $\pi : \Str \to \Loc$
is a functor $\Ranpi\AA : \Loc \to \Alg$, together with a natural transformation
$\epsilon : \Ranpi\AA \circ\pi \Rightarrow \AA$, that is universal in the following sense: 
Given any functor $\BB : \Loc \to \Alg$ and natural transformation
$\zeta :  \BB\circ\pi \Rightarrow \AA$, then $\zeta$ uniquely factors through $\epsilon$.
\end{defi}
\begin{rem}
A right Kan extension may be visualized by the diagram
\begin{flalign}
\xymatrix{
\ar[dr]_-{\pi}\Str \ar[rr]^-{\AA} &\ar@{<=}[d]_-{\epsilon}& \Alg\\
&\Loc\ar[ru]_-{\Ranpi\AA}&
}
\end{flalign}
which commutes up to the natural transformation $\epsilon$. The universal property then
says that for any other such diagram
\begin{flalign}
\xymatrix{
\ar[dr]_-{\pi}\Str \ar[rr]^-{\AA} &\ar@{<=}[d]_-{\zeta}& \Alg\\
&\Loc\ar[ru]_-{\mathfrak{B}}&
}
\end{flalign}
there exists a unique natural transformation $\alpha : \mathfrak{B}\Rightarrow \Ranpi \AA$ such that
\begin{flalign}
\parbox[c]{1cm}{\xymatrix@C=4.0em@R=3.5em{
\ar[dr]_-{\pi}\Str \ar[rr]^-{\AA} &\ar@{<=}[d]_-{\zeta}& \Alg\\
&\Loc\ar[ru]_-{\mathfrak{B}}&
}} ~\quad~=~\quad~ 
\parbox[c]{1cm}{\xymatrix@C=4.0em@R=3.5em{
\ar[dr]_-{\pi}\Str \ar[rr]^-{\AA} &\ar@{<=}[d]_-{\epsilon}& \Alg\\
&\Loc\ar@/^0.75pc/[ru]^-{~~~~~~~~~~\Ranpi\AA}_-{\rotatebox[origin=c]{135}{$\Rightarrow$}\alpha~~~~~} \ar@/_0.8pc/[ru]_-{\mathfrak{B}}&
}}
\end{flalign}
If it exists, a right Kan extension is unique up to a unique isomorphism, hence it is justified to 
speak of {\em the} right Kan extension $\Ranpi\AA: \Loc\to\Alg$ 
of $\AA : \Str \to \Alg$ along $\pi : \Str \to \Loc$.
\end{rem}

Because the category $\Alg$ is complete, i.e.\ all limits in $\Alg$ exist, the
right Kan extension $\Ranpi\AA: \Loc\to\Alg$ 
of $\AA : \Str \to \Alg$ along $\pi : \Str \to \Loc$ exists and 
may be computed via limits.
For an object $M$ in $\Loc$, we denote its under-category by $M\downarrow \pi$:
Objects in $M\downarrow \pi$ are pairs $(S,h)$ consisting of an object
$S$ in $\Str$ and a $\Loc$-morphism $h : M\to \pi(S)$. Morphisms
$g : (S,h) \to (\widetilde{S}, \widetilde{h})$ in $M\downarrow \pi$ 
are $\Str$-morphisms $g:S\to\widetilde{S}$ such that the diagram
\begin{flalign}
\xymatrix{
&\ar[dl]_-{h}M\ar[dr]^-{\widetilde{h}}&\\
\pi(S) \ar[rr]_-{\pi(g)} && \pi(\widetilde{S})
}
\end{flalign}
commutes. There is a projection functor
\begin{flalign}
\QQ^M : M\downarrow \pi \longrightarrow \Str
\end{flalign}
that acts on objects as $(S,h) \mapsto S$
and on morphisms as $(g : (S,h) \to (\widetilde{S}, \widetilde{h})) \mapsto (g:S\to \widetilde{S})$.
Moreover, given any $\Loc$-morphism $f : M\to M^\prime$, there is a functor
$f\downarrow \pi : M^\prime \downarrow \pi \to M\downarrow\pi$
that acts on objects $(S^\prime,h^\prime)$ as
\begin{subequations}
\begin{flalign}
f\downarrow \pi (S^\prime,h^\prime) := (S^\prime, h^\prime\circ f)
\end{flalign}
and on morphisms $g^\prime : (S^\prime,h^\prime) \to (\widetilde{S^\prime}, \widetilde{h^\prime})$ as
\begin{flalign}
f\downarrow \pi (g^\prime) := g^\prime :  (S^\prime, h^\prime\circ f) \longrightarrow 
(\widetilde{S^\prime}, \widetilde{h^\prime}\circ f)~.
\end{flalign}
\end{subequations}
Hence, we obtain a functor
\begin{flalign}
- \downarrow \pi : \Loc^\op \longrightarrow \mathsf{Cat}
\end{flalign}
to the category $\mathsf{Cat}$ of categories.
\sk

For any object $M$ in $\Loc$, we define an object in $\Alg$ by forming the limit
\begin{flalign}\label{eqn:RanpiLimit}
\Ranpi\AA(M) := \lim \Big(M\downarrow \pi \stackrel{\QQ^M}{\longrightarrow}  \Str\stackrel{\AA}{\longrightarrow} \Alg\Big)
\end{flalign}
in the category $\Alg$. Given any $\Loc$-morphism $f : M\to M^\prime$, there are two functors
from $M^\prime \downarrow \pi $ to $\Alg$,
\begin{flalign}
M^\prime \downarrow \pi \stackrel{f\downarrow \pi}{\longrightarrow}  M \downarrow \pi
\stackrel{\QQ^M}{\longrightarrow}  \Str\stackrel{\AA}{\longrightarrow} \Alg\quad,\qquad 
M^\prime \downarrow \pi \stackrel{\QQ^{M^\prime}}{\longrightarrow}  \Str\stackrel{\AA}{\longrightarrow} \Alg~~,
\end{flalign}
and a natural transformation $\eta : \AA\circ \QQ^M\circ f\downarrow \pi \Rightarrow \AA\circ \QQ^{M^\prime}$
with components given by $\eta_{(S^\prime,h^\prime)}^{} =\id_{\AA(S^\prime)} : \AA(S^\prime) \to \AA(S^\prime)$, for all
objects $(S^\prime,h^\prime)$ in $M^\prime\downarrow \pi$. By universality of limits, this defines an $\Alg$-morphism
\begin{flalign}
\Ranpi\AA(f) : \Ranpi\AA(M)\longrightarrow \Ranpi\AA(M^\prime)~,
\end{flalign}
for any $\Loc$-morphism $f:M\to M^\prime$. It is easy to show that the construction above defines a functor
\begin{flalign}
\Ranpi\AA : \Loc\longrightarrow \Alg~.
\end{flalign}
Moreover, there is a natural transformation $\epsilon : \Ranpi\AA\circ \pi \Rightarrow \AA$
with components
\begin{flalign}
\epsilon_{S}^{} := \pr_{(S,\id_{\pi(S)})} : \Ranpi\AA(\pi(S)) \longrightarrow \AA(S)~,
\end{flalign}
for all objects $S$ in $\Str$, given by the canonical projections from the limit \eqref{eqn:RanpiLimit} to the factor
labeled by the  object $(S,\id_{\pi(S)})$ in $\pi(S)\downarrow \pi$.
\begin{theo}[{\cite[Theorem X.3.1]{MacLane}}]\label{theo:Kanextension}
The functor $\Ranpi\AA : \Loc\to \Alg$ together with the natural
transformation $\epsilon : \Ranpi\AA\circ \pi \Rightarrow \AA$ constructed above
is a right Kan extension of $\AA : \Str\to \Alg$ along $\pi : \Str\to\Loc$.
\end{theo}

It is instructive to provide also more explicit formulas for the right Kan extension constructed above:
Given any object $M$ in $\Loc$, the limit in \eqref{eqn:RanpiLimit} can be expressed as
\begin{flalign}\label{eqn:RanpiLimitexplicit}
\Ranpi\AA(M) = \Big\{a \in \!\!\!\prod\limits_{(S,h)\in (M\downarrow\pi)_0} \!\!\!\AA(S) ~:~
\AA(g)\big(a(S,h)\big) = a(\widetilde{S},\widetilde{h})\,,~\forall g:(S,h)\to (\widetilde{S},\widetilde{h})  \Big\}~.
\end{flalign}
In this expression we have regarded elements 
$a\in \!\!\!\prod\limits_{(S,h)\in (M\downarrow\pi)_0}\!\!\!\AA(S)$ of the product as mappings
\begin{flalign}
(M\downarrow\pi)_0 \ni (S,h) \longmapsto a(S,h) \in \AA(S)
\end{flalign}
from the objects $(M\downarrow\pi)_0$ of the category $M\downarrow\pi$
to the functor $\AA\circ\QQ^M$.
Given any $\Loc$-morphism $f : M\to M^\prime$, the $\Alg$-morphism
$\Ranpi\AA(f) : \Ranpi\AA(M)\to\Ranpi\AA(M^\prime)$ maps an element
$a \in \Ranpi\AA(M)$ to the element in $\Ranpi\AA(M^\prime)$ specified by 
\begin{flalign}\label{eqn:Ranpimorphexplicit}
\big(\Ranpi\AA(f)(a)\big)(S^\prime,h^\prime) = a(S^\prime,h^\prime\circ f)~,
\end{flalign}
for all objects $(S^\prime,h^\prime)$ in $M^\prime\downarrow \pi$.
Finally, the natural transformation $\epsilon : \Ranpi\circ \pi \Rightarrow \AA$ has components
\begin{flalign}
\epsilon_S^{} : \Ranpi\AA(\pi(S)) \longrightarrow  \AA(S)~,~~ a\longmapsto a(S, \id_{\pi(S)})~,
\end{flalign}
for all objects $S$ in $\Str$. The fact that $\epsilon$ is a natural transformation, i.e.\ that
for any $\Str$-morphism $g : S\to S^\prime$ the diagram
\begin{flalign}
\xymatrix{
\ar[d]_-{\epsilon_S^{}}\Ranpi\AA(\pi(S)) \ar[rr]^-{\Ranpi\AA(\pi(g))} && \Ranpi\AA(\pi(S^\prime))\ar[d]^-{\epsilon_{S^\prime}^{}}\\
\AA(S) \ar[rr]_-{\AA(g)} && \AA(S^\prime)
}
\end{flalign}
in $\Alg$ commutes, may also be confirmed by an explicit computation: Mapping an element
$a\in \Ranpi\AA(\pi(S))$ along the upper path of this diagram, we obtain
\begin{subequations}
\begin{flalign}
\epsilon_{S^\prime}^{}\big(\Ranpi\AA(\pi(g))(a)\big) = 
\big(\Ranpi\AA(\pi(g))(a)\big)(S^\prime,\id_{\pi(S^\prime)})
= a(S^\prime,\pi(g))~,
\end{flalign}
while going the lower path we obtain
\begin{flalign}
\AA(g)\big(\epsilon_S^{}(a)\big) = \AA(g)\big(a(S,\id_{\pi(S)})\big)~.
\end{flalign}
\end{subequations}
Equality holds true because of the compatibility conditions in \eqref{eqn:RanpiLimitexplicit} and the fact that
the $\Str$-morphism $g: S\to S^\prime$ defines a $\pi(S)\downarrow \pi$-morphism
$g : (S,\id_{\pi(S)}) \to (S^\prime, \pi(g))$.
\sk

At first glance, the functor $\Ranpi\AA : \Loc\to \Alg$ of Theorem \ref{theo:Kanextension},
which we later would like to interpret as a quantum field theory on $\Loc$,
seems to be non-local: For an object $M$ in $\Loc$
the algebra $\Ranpi\AA(M)$ is constructed as a limit over the under-category $M\downarrow \pi$ 
[cf.\ \eqref{eqn:RanpiLimit} and \eqref{eqn:RanpiLimitexplicit}] and hence it seems to
depend on algebras $\AA(S)$ associated to 
structured spacetimes $S$ whose underlying spacetime $\pi(S)$ is larger than $M$.
Using that, by assumption, $\pi :\Str\to \Loc$ is a category fibered in groupoids,
we shall now show that $\Ranpi\AA(M)$ is isomorphic to a limit over the groupoid
$\pi^{-1}(M)$ and hence just depends on the algebras $\AA(S)$ for structured spacetimes
$S$ over $M$.
\begin{theo}\label{theo:initial}
If $\pi:\Str\to \Loc$ is a category fibered in groupoids (or just a fibered category), then,
for any object $M$ in $\Loc$, there exists a canonical isomorphism
\begin{flalign}
\xymatrix{
\Ranpi\AA(M) ~\ar[r]^-{\simeq}& ~\lim \AA\vert_{\pi^{-1}(M)}^{}
}~~,
\end{flalign}
where $\AA\vert_{\pi^{-1}(M)} : \pi^{-1}(M)\to\Alg$ 
denotes the restriction of the functor $\AA : \Str \to \Alg$
to the subcategory $\pi^{-1}(M)$ of $\Str$. 
\end{theo}
\begin{proof}
Let $M$ be any object in $\Loc$. Then there exists a functor $\iota : \pi^{-1}(M) \to M\downarrow \pi$
which assigns to an object $S$ in $\pi^{-1}(M)$
the object $\iota(S) := (S, \id_M)$ and to a $\pi^{-1}(M)$-morphism $g : S\to S^\prime$
the $M\downarrow\pi$-morphism
$\iota(g) = g : (S,\id_M) \to (S^\prime,\id_M)$. Notice that
\begin{flalign}
\AA\circ\QQ^M\circ \iota = \AA\vert_{\pi^{-1}(M)}
\end{flalign}
as functors from $\pi^{-1}(M)$ to $\Alg$. Hence, there exists a canonical $\Alg$-morphism
\begin{flalign}
\xymatrix{
\Ranpi\AA(M) = \lim \Big(M\downarrow \pi \stackrel{\QQ^M}{\longrightarrow}  \Str\stackrel{\AA}{\longrightarrow} \Alg\Big) 
 ~\ar[r]& ~\lim\Big( \AA\vert_{\pi^{-1}(M)}^{} : \pi^{-1}(M)\to \Alg  \Big)
}~.
\end{flalign}
Our claim that this morphism is an isomorphism would follow if we can show that
the functor $\iota : \pi^{-1}(M)\to M\downarrow\pi$ is initial, i.e.\
the dual notion of final, cf.\ \cite[Chapter IX.3]{MacLane}.
This is the goal of the rest of this proof.
The functor $\iota: \pi^{-1}(M)\to M\downarrow\pi$ is by definition {\em initial} if the following properties hold true:
\begin{enumerate}
\item For all objects $(S,h)$ in $M\downarrow \pi$ there exists
an object $S^\prime$ in $\pi^{-1}(M)$ and an $M\downarrow \pi$-morphism $\iota(S^\prime) \to (S,h)$;
\item For any object $(S,h)$ in $M\downarrow \pi$,  any objects $S^\prime,S^{\prime\prime}$ in $\pi^{-1}(M)$
and any $M\downarrow \pi$-diagram $\iota(S^\prime) \to (S,h)\leftarrow \iota(S^{\prime\prime}) $,
there exists a zig-zag of $\pi^{-1}(M)$-morphisms
\begin{flalign}
S^\prime = S_0 \leftarrow S_1 \rightarrow S_2 \leftarrow S_3\rightarrow \cdots \rightarrow S_{2n} =S^{\prime\prime}
\end{flalign}
and $M\downarrow\pi$-morphisms $\iota(S_i) \to (S,h)$, for $i=0,\dots,2n$, such that the diagrams
\begin{flalign}
\xymatrix{
\ar[dr]\iota(S_{2k}) & \ar[r]\ar[d]\iota(S_{2k+1}) \ar[l]& \ar[dl]\iota(S_{2k+2})\\
&(S,h)&
}
\end{flalign}
commute, for all $k=0,\dots,n-1$.
\end{enumerate}
We  show the first property: Let $(S,h)$ be any object in $M\downarrow \pi$.
By Definition \ref{def:fiberedcategory}, there exists a (cartesian) $\Str$-morphism
$g : S^\prime \to S$ such that $\pi(g) = h : M\to \pi(S)$. It follows that 
$S^\prime$ is an object in $\pi^{-1}(M)$ and that $g : \iota(S^\prime) = (S^\prime,\id_M) \to (S,h)$
is an $M\downarrow \pi$-morphism.
\sk

We next show the second property: Let $(S,h)$ be any object in $M\downarrow \pi$,  
let $S^\prime,S^{\prime\prime}$ be any objects in $\pi^{-1}(M)$
and let $\iota(S^\prime) \to (S,h)\leftarrow \iota(S^{\prime\prime}) $
be any $M\downarrow \pi$-diagram. The latter is given by two
$\Str$-morphisms $g^\prime : S^\prime \to S$ and $g^{\prime\prime}: S^{\prime\prime}\to S$
satisfying $\pi(g^\prime) = \pi(g^{\prime\prime}) = h : M\to \pi(S)$.
Using Definition \ref{def:fiberedcategory}, we further can find a cartesian
$\Str$-morphism $g : S_2 \to S$ such that $\pi(g) = h : M\to\pi(S)$. 
Notice that $S_2$ is an object in $\pi^{-1}(M)$ and that
$g : \iota(S_2) \to (S,h)$ is an $M\downarrow\pi$-morphism.
Moreover, we obtain the commutative diagram
\begin{subequations}
\begin{flalign}
\xymatrix{
\ar[dr]_-{\id_M}\pi(S^\prime)=M \ar[rr]^-{\pi(g^\prime) =h } && \pi(S)\\
&\pi(S_2) =M\ar[ru]_-{\pi(g) =h}&
}
\end{flalign}
in $\Loc$ and the incomplete diagram
\begin{flalign}
\xymatrix{
S^\prime \ar[rr]^-{g^\prime} && S\\
&S_2\ar[ru]_-{g }&
}
\end{flalign}
\end{subequations}
in $\Str$. By Definition \ref{def:cartesian}, there exists a unique
$\Str$-morphism  $\widetilde{g^\prime} : S^\prime \to S_2$ completing this diagram
to a commutative diagram, such that $\pi(\widetilde{g^\prime}) =\id_M$.
Hence, $\widetilde{g^\prime}$ is a $\pi^{-1}(M)$-morphism. 
Replacing $S^\prime$ by $S^{\prime\prime}$ in this construction,
we obtain a $\pi^{-1}(M)$-morphism $\widetilde{g^{\prime\prime}} : S^{\prime\prime}\to S_2$.
This yields the following zig-zag of $\pi^{-1}(M)$-morphisms
\begin{flalign}
\xymatrix{
S^\prime & \ar[l]_-{\id_{S^\prime}} S^\prime  \ar[r]^-{\widetilde{g^\prime}} &S_{2} & \ar[l]_-{\widetilde{g^{\prime\prime}}}  S^{\prime\prime}\ar[r]^-{\id_{S^{\prime\prime}}} & S^{\prime\prime}
}~,
\end{flalign}
and the associated diagrams
\begin{flalign}
\xymatrix{
\ar[dr]_-{g^\prime}\iota(S^\prime) &\ar[d]_-{g^\prime} \ar[l]_-{\id_{S^\prime}} \iota(S^\prime)  \ar[r]^-{\widetilde{g^\prime}} 
& \ar[dl]^-{g}\iota(S_{2}) &\ar[dr]_-{g}\iota(S_{2})& \ar[d]_-{g^{\prime\prime}}\ar[l]_-{\widetilde{g^{\prime\prime}}}  \iota(S^{\prime\prime})\ar[r]^-{\id_{S^{\prime\prime}}} & \ar[dl]^-{g^{\prime\prime}} \iota(S^{\prime\prime})\\
&(S,h)&& &(S,h)&
}
\end{flalign}
in $M\downarrow \pi$ commute by construction.
\end{proof}

Let us make this isomorphism more explicit:
Given any object $M$ in $\Loc$,
the limit of the restricted functor $\AA\vert_{\pi^{-1}(M)} : \pi^{-1}(M) \to \Alg$
is given by
\begin{flalign}\label{eqn:limArestrexplicit}
\lim\AA\vert_{\pi^{-1}(M)} = \Big\{a \in \!\!\!\prod\limits_{S\in \pi^{-1}(M)_0} \!\!\!\AA(S) ~:~
\AA(g)\big(a(S)\big) = a(\widetilde{S})\,,~\forall g: S \to \widetilde{S}  \Big\}~,
\end{flalign}
where we again regard elements of the product as mappings 
\begin{flalign}
\pi^{-1}(M)_0\ni S\longmapsto a(S) \in \AA(S)
\end{flalign}
from the objects $\pi^{-1}(M)_0$ of the groupoid $\pi^{-1}(M)$ to the functor $\AA$.
Denoting the isomorphism established in Theorem \ref{theo:initial} by
\begin{subequations}\label{eqn:ordinarykappamap}
\begin{flalign}
\kappa_M^{} : \Ranpi\AA(M)\longrightarrow \lim\AA\vert_{\pi^{-1}(M)}~,
\end{flalign}
we explicitly have that
\begin{flalign}
\big(\kappa_M^{}(a)\big)(S) := a(S,\id_M)~,
\end{flalign}
\end{subequations}
for all $a\in\Ranpi\AA(M)$ and all objects $S$ in $\pi^{-1}(M)$. 
To find an explicit expression for the inverse of $\kappa_M^{}$, 
we shall fix for each object $S$ in $\Str$ and each $\Loc$-morphism
$f : M\to \pi(S)$ a choice of pullback $f_\ast : f^\ast S \to S$.
Technically, this is called a {\em cleavage} 
of the fibered category $\pi:\Str\to \Loc$, cf.\ \cite{Vistoli}.
We set
\begin{flalign}
\big(\kappa_M^{-1}(a)\big)(S,h) := \AA(h_\ast)\big(a(h^\ast S)\big)~,
\end{flalign}
for all $a\in\lim\AA\vert_{\pi^{-1}(M)}$ and all objects $(S,h)$ in $M\downarrow \pi$.
It is easy to verify that $\kappa_M^{-1}$ is the inverse of $\kappa_M^{}$ and hence that
it does not depend on the choice of cleavage: We have that
\begin{subequations}
\begin{flalign}
\big(\kappa_M^{}\circ \kappa_M^{-1}(a)\big)(S) = \big(\kappa_M^{-1}(a)\big)(S,\id_M) = 
\AA({\id_M}_\ast)\big(a({\id_M}^\ast S)\big) = a(S)~,
\end{flalign}
for all $a\in\lim\AA\vert_{\pi^{-1}(M)}$ and all objects $S$ in $\pi^{-1}(M)$,
where in the last equality we have used the compatibility condition in \eqref{eqn:limArestrexplicit}
for the $\pi^{-1}(M)$-morphism ${\id_M}_\ast : {\id_M}^\ast S\to S$.
Moreover, we have that
\begin{flalign}
\big(\kappa_M^{-1}\circ \kappa_M^{}(a)\big)(S,h) = \AA(h_\ast)\big(\big(\kappa_M^{}(a)\big)(h^\ast S)\big)
= \AA(h_\ast)\big(a(h^\ast S,\id_M)\big)  = a(S,h)~,
\end{flalign}
\end{subequations}
for all $a\in \Ranpi\AA(M)$ and all objects $(S,h)$ in $M\downarrow \pi$,
where in the last equality we have used the compatibility condition in \eqref{eqn:RanpiLimitexplicit}
for the $M\downarrow \pi$-morphism $h_\ast : (h^\ast S,\id_M) \to (S,h)$.
\sk

Notice that the assignment of algebras $M \mapsto \lim\AA\vert_{\pi^{-1}(M)}$ 
does {\em not} admit an obvious functorial structure because the assignment of
groupoids $M\mapsto \pi^{-1}(M)$, which determines the shape of the diagrams $\AA\vert_{\pi^{-1}(M)}$, 
is only pseudo-functorial (after choosing any cleavage), cf.\ \cite{Vistoli}. We may however make use of
our canonical isomorphisms $\kappa_M^{}$ given in \eqref{eqn:ordinarykappamap}
in order to {\em equip} the assignment of algebras  $M \mapsto \lim\AA\vert_{\pi^{-1}(M)}$ 
with the functorial structure induced by $\Ranpi\AA : \Loc\to\Alg$.
By construction, the $\kappa_M^{}$ will then become the components of a natural isomorphism
between $\Ranpi\AA : \Loc\to\Alg$ and a more convenient
and efficient model for the right Kan extension which is based on the assignment $M \mapsto \lim\AA\vert_{\pi^{-1}(M)}$. 
Concretely, the construction is as follows: We define the functor 
\begin{subequations}\label{eqn:UpiAAfunctor}
\begin{flalign}
\U\AA : \Loc \longrightarrow \Alg
\end{flalign}
by setting
\begin{flalign}
\U\AA (M) := \lim\AA\vert_{\pi^{-1}(M)} = \Big\{a \in \!\!\!\prod\limits_{S\in \pi^{-1}(M)_0} \!\!\!\AA(S) ~:~
\AA(g)\big(a(S)\big) = a(\widetilde{S})\,,~\forall g: S \to \widetilde{S}  \Big\}~,
\end{flalign}
for all objects $M$ in $\Loc$, and
\begin{flalign}
\U\AA(f) := \kappa_{M^\prime}^{}\circ \Ranpi\AA(f) \circ \kappa_{M}^{-1} : 
\U\AA (M) \longrightarrow \U\AA (M^\prime)~, 
\end{flalign}
\end{subequations}
for all $\Loc$-morphisms $f:M\to M^\prime$. Explicitly, the $\Alg$-morphism $\U\AA(f)$
acts as
\begin{flalign}\label{eqn:UpiAAmorphexplicit}
\big(\U\AA(f)(a)\big)(S^\prime) = \AA(f_\ast)\big(a(f^\ast S^\prime)\big)~,
\end{flalign}
for all $a\in \U\AA (M)$ and all objects $S^\prime$ in $\pi^{-1}(M^\prime)$.
By construction, we obtain
\begin{cor}\label{cor:effectiveRan}
The natural transformation $\kappa : \Ranpi\AA \Rightarrow \U\AA$ with components
given in  \eqref{eqn:ordinarykappamap} is a natural isomorphism. In particular,
the functor $\U\AA : \Loc\to \Alg$ given by \eqref{eqn:UpiAAfunctor} 
is a model for the right Kan extension of $\AA :\Str\to \Alg$ along $\pi:\Str\to \Loc$.
\end{cor}

\begin{rem}\label{rem:physicalinterpretation}
The algebra $\U\AA(M)$ assigned by our model  \eqref{eqn:UpiAAfunctor} for the right Kan extension
to a spacetime $M$ has the following natural physical interpretation: Motivated by our examples
in Section \ref{sec:setup}, $\pi^{-1}(M)$ should be interpreted as the groupoid 
of geometric structures and their gauge transformations over $M$. Elements in $\U\AA(M)$ are
then functions on the space of geometric structures $\pi^{-1}(M)_0$ with values in the quantum field theory
$\AA : \Str \to \Alg$ in the sense that they assign to each structure $S\in\pi^{-1}(M)_0$ 
an element $a(S) \in \AA(S)$ of its corresponding algebra. 
The compatibility conditions in  $\U\AA(M)$ (cf.\ \eqref{eqn:UpiAAfunctor})
enforce gauge-equivariance of such $\AA$-valued functions.
In other words, $\U\AA(M)$ is an algebra of observables which describes 
gauge invariant combinations of classical (i.e.\ not quantized) observables for the 
geometric structures over $M$ and quantum observables for the quantum 
field theory $\AA : \Str\to\Alg$ for each structure $S\in\pi^{-1}(M)_0$ over $M$. 
In the spirit of Doplicher, Haag and Roberts \cite{DHR1,DHR2} 
and its generalization to Lorentzian manifolds \cite{Ruzzi,BR},
a major achievement would be to understand 
which aspects of the information encoded by the groupoids $\pi^{-1}(M)$
may be reconstructed from the right Kan extension $\U\AA :\Loc\to\Alg$.
Such questions however lie beyond the scope of this paper as they require 
a careful treatment of functional analytical aspects of quantum field theories, 
as well as their states and representation theory.
\sk

In the simplest scenario where the category fibered in groupoids 
$\pi: \Loc\times G\to \Loc$ is the one corresponding to a ``global gauge group'' 
$G$ (cf.\ Example \ref{ex:globalgauge}), the algebra $\U\AA(M)$
is just the fixed-point subalgebra of the $G$-action on $\AA(M)$. 
This agrees with Fewster's construction in \cite{FewsterAuto}.
For the particular case of Kaluza-Klein theories (cf.\ Remark \ref{rem:KK}),
$\U\AA(M)$ is the subalgebra of the algebra $\AA(M\times K)$ 
of an $m+k$-dimensional quantum field theory consisting of those
observables that are invariant under the isometries of $K$. If for example 
$K = \bbT^k$ is the usual flat $k$-torus, then the $\mathrm{Iso}^+(K)$-invariant
observables are those observables on $M\times K$ that
``have zero momentum along the extra-dimensions''. From the $m$-dimensional 
perspective, momentum along the extra-dimensions 
corresponds to additional mass terms, hence one may interpret 
such observables as ``low-energy'' observables arising from a Kaluza-Klein reduction.
\end{rem}

\section{\label{sec:properties}Properties}
A natural question to ask is whether the right Kan extension 
$\U\AA : \Loc\to \Alg$ of Corollary \ref{cor:effectiveRan} 
is an ordinary quantum field theory in the sense of \cite{Brunetti},
i.e.\ a quantum field theory on $\id_{\Loc}:\Loc\to\Loc$ according to Definition \ref{def:QFT}.
In general, the answer will be negative, unless the category fibered in groupoids $\pi : \Str\to\Loc$
satisfies suitable extension properties, mimicking for example the notion of a deterministic
time evolution. (See Remark \ref{rem:flabbyinterpretation} below for a more detailed 
physical interpretation of the following definition.) 
\begin{defi}\label{def:flabby}
\begin{itemize}
\item[a)] A category fibered in groupoids $\pi : \Str \to \Loc$ is called {\em flabby}
if for every object $S$ in $\Str$ and every $\Loc$-morphism $f : \pi(S)\to M^\prime$  
there exists a $\Str$-morphism $g : S\to S^\prime$ with the property $\pi(g) = f : \pi(S)\to M^\prime$.

\item[b)] A category fibered in groupoids $\pi : \Str \to \Loc$ is called {\em Cauchy flabby}
if the following two properties hold true: 1.)~For every object $S$ in $\Str$ and 
every Cauchy $\Loc$-morphism $f : \pi(S)\to M^\prime$  
there exists a $\Str$-morphism $g : S\to S^\prime$ such that $\pi(g) = f : \pi(S)\to M^\prime$.
2.)~Given two such $\Str$-morphisms $g:S\to S^\prime$ and $\widetilde{g}:S\to \widetilde{S^{\prime}}$,
there exists a $\pi^{-1}(M^\prime)$-morphism 
$g^\prime : S^\prime\to \widetilde{S^\prime}$ such that the $\Str$-diagram
\begin{flalign}\label{eqn:Cauchyflabbydiagram}
\xymatrix{
&\ar[ld]_-{g}S\ar[rd]^-{\widetilde{g}}&\\
S^\prime \ar[rr]_-{g^\prime}& & \widetilde{S}^\prime
}
\end{flalign}
commutes.
\end{itemize}
\end{defi}
\begin{rem}\label{rem:flabbyinterpretation}
Loosely speaking, the flabbiness condition formalizes the idea that the geometric structures
on $\Loc$ which are described by the category $\Str$ always 
admit an extension (possibly non-canonical) from smaller to larger spacetimes.
This is in analogy to the flabbiness condition in sheaf theory. As we shall see later,
flabbiness is a very restrictive condition that {\em is not} satisfied in our examples
of categories fibered in groupoids presented in Section \ref{sec:setup}. 
(It holds only for the very special case describing ``global gauge groups'', see Example
\ref{ex:globalgauge} and also Example \ref{ex:globalgauge2} below.)
On the other hand, the Cauchy flabbiness condition formalizes the idea
that the geometric structures on $\Loc$ which are described by the category $\Str$
admit a `time evolution' that is unique up to isomorphisms. As we shall clarify
in the examples at the end of this section, Cauchy flabbiness {\em is} satisfied
in many of our examples of categories fibered in groupoids presented in Section \ref{sec:setup}
after performing some obvious and well-motivated modifications. 
\end{rem}
\begin{theo}\label{theo:properties}
Let $\pi: \Str\to \Loc$ be a category fibered in groupoids
and $\AA :\Str \to\Alg$ a quantum field theory in the sense of Definition \ref{def:QFT}.
Then the right Kan extension $\U\AA : \Loc\to \Alg$ (cf.\ Corollary \ref{cor:effectiveRan}) 
satisfies the causality axiom. It satisfies the isotony axiom if and only if $\pi:\Str\to\Loc$ is flabby.
If $\pi:\Str\to\Loc$ is Cauchy flabby, then $\U\AA : \Loc\to \Alg$ satisfies the time-slice axiom.
The converse of this statement is not true, see Remark \ref{rem:converse} below.
\end{theo}
\begin{proof}
{\em Causality:} Let $M_1 \stackrel{f_1}{\longrightarrow}M \stackrel{f_2}{\longleftarrow} M_2$ be a
causally disjoint $\Loc$-diagram. Take any two elements
$a\in \U\AA(M_1)$ and $b\in \U\AA(M_2)$.
Using \eqref{eqn:UpiAAmorphexplicit}, we obtain
\begin{flalign}\label{eqn:commutatortmp}
\big[\U\AA(f_1)(a),\U\AA(f_2)(b)\big](S) = 
\big[\AA({f_1}_\ast) \big( a(f_1^\ast S) \big), \AA({f_2}_\ast)\big( b(f_2^\ast S) \big)\big]~,
\end{flalign}
for all objects $S$ in $\pi^{-1}(M)$, where the first commutator is in $\U\AA(M)$ and the second one is in $\AA(S)$. 
It then follows that \eqref{eqn:commutatortmp}
is equal to zero because $f_1^\ast S \stackrel{{f_1}_\ast}{\longrightarrow}S \stackrel{{f_2}_\ast}{\longleftarrow} f_2^\ast S$ 
projects via $\pi$ to our causally disjoint $\Loc$-diagram and $\AA:\Str\to\Alg$ satisfies 
by assumption the causality axiom. Hence, $\U\AA$ satisfies causality.
\sk

{\em Isotony:} 
Let us first prove the direction ``$\pi : \Str\to\Loc$ is flabby'' $\Rightarrow$ ``$\U\AA$ satisfies isotony'':
Let $f : M\to M^\prime$ be any $\Loc$-morphism. 
Let $a \in\U\AA(M)$ be any element such that $\U\AA(f)(a)=0$,
i.e.\ $\AA(f_\ast)(a(f^\ast S^\prime)) =0 $, for all objects $S^\prime$ in $\pi^{-1}(M^\prime)$. 
Using isotony of $\AA : \Str\to\Alg$, this implies that
$a(f^\ast S^\prime) =0 $, for all objects $S^\prime$ in $\pi^{-1}(M^\prime)$.
We have to show that $a(S) =0$, for all objects $S$ in $\pi^{-1}(M)$.
By the flabbiness assumption, there exists for every object $S$ in $\pi^{-1}(M)$ a
$\Str$-morphism $g : S\to S^\prime$  such that $\pi(g) = f: M\to M^\prime$.
As a consequence, $S$ is isomorphic (in $\pi^{-1}(M)$) to any pullback
$f^\ast S^\prime$ of $S^\prime$ to $M$ along $f$. 
(In fact, since $f_\ast: f^\ast S^\prime \to S^\prime$ is cartesian, 
we get a unique $\pi^{-1}(M)$-morphism $S \to f^\ast S^\prime$
that can be inverted as $\pi^{-1}(M)$ is a groupoid.) Choosing such a
$\pi^{-1}(M)$-morphism $\widetilde{g} : f^\ast S^\prime \to S$,
the compatibility conditions in \eqref{eqn:UpiAAfunctor} imply that
$\AA(\widetilde{g})(a(f^\ast S^\prime))=a(S)$ and thus
$a(S) =0$ because we have seen above that $a(f^\ast S^\prime) =0 $.
Hence, $\U\AA$ satisfies isotony.
\sk

We prove the opposite direction by contraposition, i.e.\ 
``$\pi:\Str \to\Loc$ is not flabby'' $\Rightarrow$ ``$\U\AA$ does not satisfy isotony'':
As $\pi:\Str\to\Loc$ is not flabby, we can find an object $S$ in $\Str$ and a $\Loc$-morphism
$f:M \to M^\prime$ (with $M=\pi(S)$) such that there exists {\em no} $\Str$-morphism $g: S\to S^\prime$
with the property $\pi(g) =f :M\to M^\prime$. Let $a\in \U\AA(M)$ be the element 
specified by
\begin{flalign}\label{eqn:adeftmp}
a(\widetilde{S}) = 
\begin{cases}
\1 & \,,~\text{if $\widetilde{S}\simeq S$ in $\pi^{-1}(M)$}~,\\
0 &\,,~\text{else}~,
\end{cases}
\end{flalign}
for all objects $\widetilde{S}$ in $\pi^{-1}(M)$, where $\1$ is the unit element in $\AA(\widetilde{S})$.
It follows that $a\neq 0$ and 
\begin{flalign}
\big(\U\AA(f)(a)\big)(S^\prime) = \AA(f_\ast)\big(a(f^\ast S^\prime)\big) =0~,
\end{flalign}
for all objects $S^\prime$ in $\pi^{-1}(M^\prime)$. The latter statement
is a consequence of $f^\ast S^\prime \not\simeq S$ 
(otherwise $S^\prime$ would be an extension of $S$ along $f$, which is against the hypothesis) 
and of our particular choice \eqref{eqn:adeftmp}
of the element $a\in  \U\AA(M)$.
Hence, isotony is violated.
\sk

{\em Time-slice:} Let $\pi : \Str\to\Loc$ be Cauchy flabby.
If $f: M\to M^\prime$ is a Cauchy $\Loc$-morphism,
a similar argument as in the proof of isotony above shows that $\U\AA(f) : 
\U\AA(M)\to \U\AA(M^\prime)$ is
injective. It hence remains to show that this $\Alg$-morphism is surjective. 
For this let $b\in \U\AA(M^\prime)$ be an arbitrary element. 
We have to find a preimage, i.e.\ an element $a\in \U\AA(M)$ 
such that $\U\AA(f)(a) =b$. 
Given any object $S$ in $\pi^{-1}(M)$, Cauchy flabbiness allows us to 
find an extension $g : S\to S^\prime$ such that $\pi(g) = f: M\to M^\prime$.
Making an arbitrary choice of such extensions, we set
\begin{flalign}\label{eqn:aSdeftmp}
a(S) := \AA(g)^{-1}\big(b(S^\prime)\big)~,
\end{flalign}
for all objects $S$ in $\pi^{-1}(M)$, where we also have
used that $\AA : \Str\to\Alg$ satisfies the
time-slice axiom in order to define the inverse $\AA(g)^{-1}$.
Notice that \eqref{eqn:aSdeftmp} does not depend on the choice of extension:
Given any other extension $\widetilde{g} : S\to\widetilde{ S^\prime}$, there exists by definition of
Cauchy flabbiness a $\pi^{-1}(M^\prime)$-morphism $g^\prime: S^\prime\to\widetilde{S^\prime}$
such that $\widetilde{g} = g^\prime\circ g$ and hence
\begin{flalign}
\AA(\widetilde{g})^{-1}\big(b(\widetilde{S^\prime})\big) = 
\AA(g)^{-1}\circ \AA(g^{\prime\,-1}) \big(b(\widetilde{S^\prime})\big)
=\AA(g)^{-1}\big(b(S^\prime)\big) = a(S)~.
\end{flalign}
It remains to prove that the family of elements defined in \eqref{eqn:aSdeftmp}
satisfies the compatibility conditions: For any $\pi^{-1}(M)$-morphism
$\widetilde{g} : S\to\widetilde{S}$ we have that
\begin{flalign}
\AA(\widetilde{g})\big(a(S)\big) = \AA(g\circ \widetilde{g}^{-1})^{-1}\big(b(S^\prime)\big)=a(\widetilde{S})~,
\end{flalign}
where in the last step we used that \eqref{eqn:aSdeftmp} does not depend on the choice of extension.
By a similar argument, we can confirm that $a$ is indeed a preimage of $b$,
\begin{flalign}
\U\AA(f)(a)(S^\prime) = \AA(f_\ast)\big(a(f^\ast S^\prime)\big) = 
\AA(f_\ast) \circ\AA(f_\ast)^{-1}\big(b(S^\prime)\big) = b(S^\prime)~,
\end{flalign}
where in the second step we used  \eqref{eqn:aSdeftmp} and that $f_\ast : f^\ast S^\prime\to S^{\prime}$ is an extension.
Hence, $\U\AA$ satisfies time-slice.
\end{proof}
\begin{rem}\label{rem:converse}
The converse implication ``$\U\AA$ satisfies time-slice'' $\Rightarrow$ ``$\pi:\Str\to\Loc$ is Cauchy flabby''
is not necessarily true. 
Consider for example $\AA = \BB \circ \pi : \Str\to\Alg$, where $\BB : \Loc \to \Alg$
is an ordinary locally covariant quantum field theory. (Physically, such theories
may be interpreted as quantum field theories that are insensitive, i.e.\ do not couple to, the
gauge theoretic structures captured by the groupoids $\pi^{-1}(M)$.) For any object $M$ in $\Loc$, we then have that
$\U\AA(M) \simeq \prod_{[S]\in \pi_0(\pi^{-1}(M))} \BB(M)$,
where $\pi_0(\pi^{-1}(M))$ denotes the set of connected components of the groupoid $\pi^{-1}(M)$, i.e.\
the quotient of the set of objects $\pi^{-1}(M)_0$ of the groupoid by the equivalence relation induced by
its morphisms. It is clear that for this particular example $\U\AA$ satisfies the time-slice axiom if and only if 
all equivalence classes $[S]\in \pi_0(\pi^{-1}(M))$ admit a unique extension along all
Cauchy $\Loc$-morphisms $f:M\to M^\prime$. This is however a weaker condition than Cauchy flabbiness,
because it does not require that for any two extensions there exists a commutative
diagram as in \eqref{eqn:Cauchyflabbydiagram}. The examples below further clarify the difference between
Cauchy flabbiness and the unique extendability of isomorphism classes along Cauchy $\Loc$-morphism.
\sk

As a side remark, it is easy to show (by contraposition) that for any quantum field theory
$\AA :\Str\to\Alg$ on a category fibered in groupoids $\pi:\Str\to\Loc$, the right Kan extension
$\U\AA:\Loc\to \Alg$ satisfies the time-slice axiom only if, for all objects $M$ in $\Loc$,
all equivalence classes $[S]\in \pi_0(\pi^{-1}(M))$ admit a unique extension along all
Cauchy $\Loc$-morphisms $f:M\to M^\prime$. 
In fact, $\U\AA(f)$ is not injective if there exists an equivalence class that cannot be extended, 
while surjectivity fails if such extension is not unique. 
\end{rem}

\begin{ex}\label{ex:spin2}
Recall the category fibered in groupoids $\pi: \mathsf{SLoc}\to\Loc$ presented in Example \ref{ex:spin}. 
Using the classification of spin structures, we will show that 
there are counterexamples to flabbiness, while Cauchy flabbiness holds true. 
Note that, for each object $M$ of $\Loc$ admitting a spin structure, 
i.e.\ such that the obstruction class in $H^2(M;\bbZ_2)$ vanishes, 
the set of isomorphism classes of the groupoid $\pi^{-1}(M)$ is an affine space over $H^1(M;\bbZ_2)$, 
the first cohomology group of $M$ with $\bbZ_2$-coefficients, see e.g.\ \cite{GreubPetry}. 
To exhibit a counterexample to flabbiness, let $M^\prime = \bbR^4$ be the 
$4$-dimensional Minkowski spacetime and consider its time zero
Cauchy surface $\{0\}\times\bbR^3$. Removing the $z$-axis of the Cauchy surface
$\Sigma:= \{0\}\times (\bbR^3 \setminus \{(0,0,z) : z\in\bbR\})$, we
define $M := D(\Sigma)\subseteq M^\prime$ as the Cauchy development $D(\Sigma)$ of $\Sigma$ 
in $M^\prime$. Notice that both $M$ and $M^\prime$ are objects of $\Loc$ 
and that the obvious inclusion of $M$ into $M^\prime$ 
provides a $\Loc$-morphism $f: M \to M^\prime$, see e.g.\ \cite[Lemma A.5.9]{BGP}. 
Observe further that $M$ is homotopic to $\mathbb{S}^1$
and hence the second cohomology group with $\bbZ_2$-coefficients vanishes for both $M$ and $M^\prime$.
As a consequence, both $M$ and $M^\prime$ admit a spin structure. 
Because $H^1(M^\prime;\bbZ_2) =0$, any two spin structures on 
$M^\prime$ are isomorphic. On the other hand, because $H^1(M;\bbZ_2)\simeq \bbZ_2$,
there exist non-isomorphic choices of spin structures on $M$.
As the pullback of any two spin structures on $M^\prime$
along $f:M\to M^\prime$ induces isomorphic spin structures on $M$,
we obtain that all spin structures on $M$ which are not isomorphic to
one that is obtained via pullback do not admit an extension
to $M^\prime$ along $f:M\to M^\prime$. Hence, flabbiness is violated.
\sk

We next show that Cauchy flabbiness holds true.
As Cauchy $\Loc$-morphisms are homotopy equivalences, they
induce isomorphisms in cohomology that allow us to extend any spin structure.
Given any two $\mathsf{SLoc}$-morphisms
$g : (M,P,\psi) \to (M^\prime,P^\prime,\psi^\prime)$ and
$\widetilde{g} : (M,P,\psi) \to (M^\prime, \widetilde{P^\prime},\widetilde{\psi^\prime})$
such that $\pi(g) = \pi(\widetilde{g}) = f:M\to M^\prime$
is a Cauchy $\Loc$-morphism, we have to show that there exists $g^\prime : (M^\prime,P^\prime,\psi^\prime)\to
(M^\prime, \widetilde{P^\prime},\widetilde{\psi^\prime})$ with $\pi(g^\prime)=\id_{M^\prime}^{}$
closing the commutative diagram in \eqref{eqn:Cauchyflabbydiagram}. From the cohomology isomorphism,
it follows that there exists indeed a $\pi^{-1}(M^\prime)$-morphism 
$g^{\prime\prime} : (M^\prime,P^\prime,\psi^\prime)\to (M^\prime, \widetilde{P^\prime},\widetilde{\psi^\prime})$,
however it is not guaranteed that  it closes the commutative diagram in \eqref{eqn:Cauchyflabbydiagram}. Fixing
any such $g^{\prime\prime}$, we consider the two parallel $\mathsf{SLoc}$-morphisms
\begin{flalign}
\xymatrix{
 (M,P,\psi) \ar@<1ex>[rr]^-{g}\ar@<-1ex>[rr]_-{\bar{g}:= g^{\prime\prime\,-1}\circ \widetilde{g}}&&  (M^\prime,P^\prime,\psi^\prime)~. 
}
\end{flalign}
Because $g$ and $\bar g$ are in particular principal $\mathrm{Spin}_0(1,m{-}1)$-bundle morphisms
(over the same $f:M \to M^\prime)$,
it is easy to confirm that there exists a unique function 
$s \in C^\infty(P,\mathrm{Spin}_0(1,m{-}1))^{\mathrm{eqv}}$ (equivariant under the adjoint action),
such that $\bar g (p) = g(p)\, s(p)$, for all $p\in P$. Regarding $s$ as
a principal bundle automorphism $s : P\to P$ over $\id_M$ (i.e.\ a gauge transformation), 
the above equality reads as $\bar g = g\circ s$.
If there exists $s^\prime \in C^\infty(P^\prime ,\mathrm{Spin}_0(1,m{-}1))^{\mathrm{eqv}}$
such that $g\circ s = s^\prime\circ g$, the above equality implies that
$\widetilde{g} = g^{\prime\prime}\circ s^\prime\circ g$, hence $g^\prime:= g^{\prime\prime}\circ s^\prime$
closes the commutative diagram in \eqref{eqn:Cauchyflabbydiagram}. In order to prove existence of
such $s^\prime$, we have to make use of the fact that $s : P\to P$ is not only a gauge transformation,
but also a $\mathsf{SLoc}$-automorphism. Its compatibility with the equivariant bundle
map $\psi: P\to FM$ to the frame bundle, i.e.\ $\psi\circ s = \psi$,
implies that the corresponding equivariant 
function $s\in C^\infty(P,\mathrm{Spin}_0(1,m{-}1))^{\mathrm{eqv}} $ takes values in the kernel
of the double covering group homomorphism $\rho: \mathrm{Spin}_0(1,m{-}1)\to \mathrm{SO}_0(1,m{-}1)$,
which is isomorphic to the group $\bbZ_2$. 
As the kernel of $\rho$ lies in the center of $\mathrm{Spin}_0(1,m{-}1)$, the adjoint action
becomes trivial, which provides us with a canonical 
isomorphism $C^\infty(P,\ker(\rho))^\mathrm{eqv} \simeq  C^\infty(M,\bbZ_2)$. 
Regarding $s$ as an element in $C^\infty(M,\bbZ_2)$, we can uniquely extend
it along $f:M\to M^\prime$ to an element 
$s^\prime\in C^\infty(M^\prime,\bbZ_2)$ because $\bbZ_2$ is a discrete group
and the image of every Cauchy morphism $f : M\to M^\prime$ 
intersects non-trivially all connected components of $M^\prime$.
Regarding this $s^\prime$ as an element in $C^\infty(P^\prime,\ker(\rho))^\mathrm{eqv} $
via the canonical isomorphism completes our proof of Cauchy flabbiness.
\end{ex}

\begin{ex}
The category fibered in groupoids $\pi: \mathsf{B}G\Loc\to\Loc$ 
presented in Example \ref{ex:gauge} is in general neither flabby nor Cauchy flabby.
Regarding flabbiness,  let us consider for example $G = U(1)$, 
$M^\prime=\bbR^4$ the $4$-dimensional Minkowski spacetime and $M:= M^\prime\setminus J_{M^\prime}(\{0\})$
the Minkowski spacetime with the closed light-cone of the origin removed. The obvious submanifold embedding
defines a $\Loc$-morphism $f : M\to M^\prime$. As $M$ is homotopic to the $2$-sphere $\mathbb{S}^2$,
principal $U(1)$-bundles over $M$ are classified up to isomorphism by the magnetic monopole charge in 
$H^2(M;\bbZ)\simeq H^2(\mathbb{S}^2;\bbZ)\simeq \bbZ$, while each principal $U(1)$-bundle over $M^\prime$
is isomorphic to the trivial one (i.e.\ charge $0$). If a principal $U(1)$-bundle $P$ over $M$
has an extension to $M^\prime$ along $f: M\to M^\prime$, i.e.\ there is a principal $U(1)$-bundle
morphism $g : P\to P^\prime$ covering $f$, then $P$ has to be isomorphic to the pullback bundle
$f^\ast P^\prime$ and in particular its monopole charge has to be $0$. This shows that 
only topologically trivial principal $U(1)$-bundles
extend for our choice of $\Loc$-morphism and hence that $\pi: \mathsf{B}G\Loc\to\Loc$ is in general not flabby.
\sk

Regarding Cauchy flabbiness, notice that
principal $G$-bundles may be extended along Cauchy $\Loc$-morphisms 
because these are homotopy equivalences
and the classification of principal $G$-bundles up to isomorphism
only depends on the homotopy type of the base manifold.
However, there exist extensions $g: (M,P)\to (M^\prime,P^\prime)$
and $\widetilde{g} : (M,P) \to (M^\prime,\widetilde{P^\prime}) $ along Cauchy $\Loc$-morphisms $f:M\to M^\prime$
for which one cannot close the commutative diagram in \eqref{eqn:Cauchyflabbydiagram}.
For example, take $G$ any Lie group of dimension $\mathrm{dim}(G) \geq 1$, 
$M^\prime = \bbR^m$ the $m$-dimensional Minkowski spacetime
and  $M= (-1,1)\times\bbR^{m-1} \subsetneqq M^\prime$. Consider the two $\mathsf{B}G\Loc$-morphisms
(between trivial principal $G$-bundles)
\begin{subequations}
\begin{flalign}
&g : M\times G \longrightarrow M^\prime\times G~,~~(x,q)\longmapsto (x,q)~,\\
&\widetilde{g} : M\times G \longrightarrow M^\prime\times G~,~~(x,q)\longmapsto (x,s(x)\,q)~,
\end{flalign}
\end{subequations}
where $s\in C^\infty(M,G)$ is any $G$-valued smooth function on $M$ which does not admit an extension to $M^\prime$.
(Because $G$ is by assumption of dimension $\geq 1$,
examples of such $s$ are functions which wildly oscillate whenever the time coordinate approaches the boundaries 
of $(-1,1)$.) Closing the commutative diagram in \eqref{eqn:Cauchyflabbydiagram} requires
an extension of $s$ to $M^\prime$, which by construction does not exist. 
Hence, $\pi: \mathsf{B}G\Loc\to\Loc$ is in general not Cauchy flabby.
As a side remark, notice that
in the case where $G$ is a discrete group (i.e.\ a $0$-dimensional Lie group), the 
category fibered in groupoids $\pi: \mathsf{B}G\Loc\to\Loc$ is Cauchy flabby. (This is similar to the
the previous example of spin structures.)
\end{ex}

\begin{ex}\label{ex:gaugecon2}
The category fibered in groupoids $\pi: \mathsf{B}G^{\mathrm{con}}\Loc\to\Loc$  presented in
Example \ref{ex:gauge} is in general neither flabby nor Cauchy flabby. The violation of flabbiness 
is because, similarly to the case of $\mathsf{B}G\Loc$ before,  principal $G$-bundles do not always 
extend along $\Loc$-morphisms and moreover, even if they would do, 
arbitrary connections do not extend as well. Cauchy flabbiness is violated because
arbitrary connections do not extend along Cauchy $\Loc$-morphisms
and, even if they would do, this extension is not unique up to isomorphism.
\sk

Let us define the full subcategory $\mathsf{B}G^{\mathrm{con}}\Loc_{\mathrm{YM}}^{}$
of $\mathsf{B}G^{\mathrm{con}}\Loc$ whose objects $(M,P,A)$ satisfy the 
Yang-Mills equation on $M$. The corresponding
category fibered in groupoids $\pi: \mathsf{B}G^{\mathrm{con}}\Loc_{\mathrm{YM}}^{} \to
\Loc$ is not flabby (by the same arguments as above) however, under certain conditions to be explained below,
it is Cauchy flabby. Assuming that the global Yang-Mills Cauchy problem is well-posed
for gauge equivalence classes, as it is the case in dimension $m=2,3,4$ and for the usual choices of  
structure group $G$ \cite{CBYM,CSYM}, any object $(M,P,A)$ in 
$\mathsf{B}G^{\mathrm{con}}\Loc_{\mathrm{YM}}^{}$ admits an extension
along any Cauchy $\Loc$-morphism $f:M\to M^\prime$ by solving the Cauchy problem.
It remains to study if, given two $\mathsf{B}G^{\mathrm{con}}\Loc_{\mathrm{YM}}^{}$-morphisms
$g : (M,P,A) \to (M^\prime,P^\prime,A^\prime)$ and
$\widetilde{g} : (M,P,A)\to(M^\prime,\widetilde{P^\prime},\widetilde{A^\prime})$ such that
$\pi(g) = \pi(\widetilde{g}) =f:M\to M^\prime$ is a Cauchy $\Loc$-morphism,
there exists a $\pi^{-1}(M^\prime)$-morphism
$g^\prime : (M^\prime,P^\prime,A^\prime) \to (M^\prime,\widetilde{P^\prime},\widetilde{A^\prime})$
closing the commutative diagram in \eqref{eqn:Cauchyflabbydiagram}.
From the well-posed Cauchy problem for equivalence classes, it follows that
there exists indeed a $\pi^{-1}(M^\prime)$-morphism 
$g^{\prime\prime} : (M^\prime,P^\prime,A^\prime) \to (M^\prime,\widetilde{P^\prime},\widetilde{A^\prime})$,
however it is not guaranteed that it closes the commutative diagram in \eqref{eqn:Cauchyflabbydiagram}.
The next steps are similar to Example \ref{ex:spin2}. We fix any such $g^{\prime\prime}$ and consider
the two parallel $\mathsf{B}G^{\mathrm{con}}\Loc_{\mathrm{YM}}^{}$-morphisms
\begin{flalign}
\xymatrix{
 (M,P,A) \ar@<1ex>[rr]^-{g}\ar@<-1ex>[rr]_-{\bar{g}:= g^{\prime\prime\,-1}\circ \widetilde{g}}&&  (M^\prime,P^\prime,A^\prime)~. 
}
\end{flalign}
There exists a unique $s\in C^\infty(P,G)^{\mathrm{eqv}}$ (i.e.\ a gauge transformation)
such that $\bar g = g\circ s$. We can close the commutative diagram 
in \eqref{eqn:Cauchyflabbydiagram} if we can find $s^\prime \in C^\infty(P^\prime,G)^{\mathrm{eqv}}$
satisfying $s^\prime\circ g = g\circ s$.
Because $s$ and $s^\prime$ are by assumption $\mathsf{B}G^{\mathrm{con}}\Loc_{\mathrm{YM}}^{}$-automorphisms,
they have to stabilize the corresponding connections $A$ and $A^\prime$. Assuming in the following
that $G$ is a matrix Lie group, the stabilizing property is equivalent to the partial differential equations (PDEs)
$\dd_{A^{(\prime)}} s^{(\prime)} := \dd s^{(\prime)} + [A^{(\prime)},s^{(\prime)}]=0$, where we
regard $s^{(\prime)}$ as an element in $\Gamma^\infty(M^{(\prime)},P^{(\prime)}\times_{\mathrm{Ad}} G)$
via the standard canonical isomorphism. This is an initial-value 
constraint for the sigma-model-type hyperbolic PDE 
$\delta_{A^{(\prime)}} \dd_{A^{(\prime)}} s^{(\prime)} =0$, where $\delta_{A^{(\prime)}}$ is the covariant 
codifferential, i.e.\ $\delta_{A^{(\prime)}} := \ast^{-1}\,\dd_{A^{(\prime)}}\,\ast$ with $\ast$ the Hodge operator.
Assuming that the Cauchy problem for this equation is well-posed, we obtain a unique $s^\prime$
from $s$ by solving the Cauchy problem, which implies  that 
$\pi: \mathsf{B}G^{\mathrm{con}}\Loc_{\mathrm{YM}}^{} \to \Loc$ is Cauchy flabby. 
By \cite{CBsigma}, this is a reasonable
assumption at least for low spacetime dimensions. Notice that in the special 
case where $G=U(1)$ is Abelian, solutions to $\dd_{A} s = \dd s =0$ are locally constant $U(1)$-valued
functions on $M$ and hence they admit a unique extension along any Cauchy $\Loc$-morphism
$f:M\to M^\prime$. Hence, $\pi: \mathsf{B}U(1)^{\mathrm{con}}\Loc_{\mathrm{YM}}^{} \to
\Loc$ is Cauchy flabby for any choice of spacetime dimension $m$.
\end{ex}

\begin{ex}\label{ex:frame2}
The category fibered in groupoids $\pi: \mathsf{FLoc}\to\Loc$ 
presented in Example \ref{ex:frame} is in general neither flabby nor Cauchy flabby.
A simple counterexample to flabbiness is as follows: Let $M \subsetneqq M^\prime= \bbR^2$ be 
any globally hyperbolic proper open subset of the $2$-dimensional Minkowski spacetime $M^\prime$, e.g.\ a diamond
or a Rindler wedge. Then there is an associated $\Loc$-morphism $f : M\to M^\prime$.
We take global coordinates $(t,x)$ on $\bbR^2$ in which the metric of both $M$ and $M^\prime$ 
takes the standard form $\dd t\otimes \dd t - \dd x\otimes\dd x$. (We also assume that 
the orientations and time-orientations are represented by $\dd t\wedge \dd x$ and $\dd t$.)
On $M$ we choose a global coframe $e$ of the form $e^0 =\dd t\,  \cosh q  + \dd x\, \sinh q$
and $e^1 = \dd t\, \sinh q +\dd x\, \cosh q$, where $q\in C^\infty(M)$ is any smooth function
on $M$ which goes to infinity towards the boundary of $M\subset M^\prime$. As a consequence, 
$q$ and therefore $e$ does not admit an extension to $M^\prime$ along $f:M\to M^\prime$ and
hence $\pi: \mathsf{FLoc}\to\Loc$  is not flabby. The same argument also shows that 
$\pi: \mathsf{FLoc}\to\Loc$ is not Cauchy flabby, as we could take for example 
$M= (-1,1)\times\bbR\subsetneqq  M^\prime$ as our globally hyperbolic proper open subset,
in which case $f:M\to M^\prime$ is a Cauchy $\Loc$-morphism. Obstructions of this kind
could be mildened by redefining the morphisms $g : (M,e)\to (M^\prime,e^\prime)$ 
in $\mathsf{FLoc}$ to be pairs $g=(f,\Lambda)$, where $f: M\to M^\prime$ is a smooth map
and $\Lambda\in C^\infty(M,\mathrm{SO}_0(1,m-1))$, such that
$f^\ast e^\prime = \Lambda\,e$, i.e.\ the coframes are preserved 
only up to a local Lorentz transformation. This is however against 
the perspective taken by Fewster in his study of the spin-statistics connection
\cite{Fewster1,Fewster2}.
\end{ex}

\begin{ex}\label{ex:source2}
The category fibered in groupoids $\pi: \mathsf{LocSrc}\to\Loc$ 
presented in Example \ref{ex:source} is in general neither flabby nor Cauchy flabby
because arbitrary functions $J\in C^\infty(M)$ do not admit an extension 
to $M^\prime$ along $f:M\to M^\prime$. Choosing any natural normally hyperbolic differential operator
$D : C^\infty\Rightarrow C^\infty$ (e.g.\ a Klein-Gordon operator)
and defining $\mathsf{LocSrc}_{D}^{}$ to be the full subcategory of
$\mathsf{LocSrc}$ whose objects $(M,J)$ satisfy the equation of motion $D_M J =0$,
then the category fibered in groupoids $\pi : \mathsf{LocSrc}_D^{}\to\Loc$ is Cauchy flabby
with a unique extension given by solving the Cauchy problem.
\end{ex}

\begin{ex}\label{ex:globalgauge2}
The category fibered in groupoids $\pi : \Loc\times G \to \Loc$
presented in Example \ref{ex:globalgauge} is both flabby and Cauchy flabby
because each fiber $\pi^{-1}(M) \simeq G$ is a groupoid with only one object.
Notice further that the morphism $g^\prime$ closing the
diagram \eqref{eqn:Cauchyflabbydiagram} defining Cauchy flabbiness 
is uniquely specified by this diagram. The same statements hold true for Kaluza-Klein
theories (cf.\ Remark \ref{rem:KK}) as these are special instances of the present example.
\end{ex}

\begin{rem}
As a consequence of Theorem \ref{theo:properties} and our examples presented 
above, we observe that the right Kan extension $\U\AA : \Loc\to \Alg$ 
satisfies always the causality axiom and very often also the time-slice axiom. 
However, it almost never satisfies the isotony axiom,
which by Theorem \ref{theo:properties} is equivalent to the very restrictive condition of
$\pi: \Str\to \Loc$ being flabby that is satisfied only by the very special and non-representative
case of ``global gauge transformations'', cf.\ Example \ref{ex:globalgauge2}.
A similar violation of isotony has been observed before in models of quantum gauge theories,
see e.g.\ \cite{DL,SDH,BDS,BDHS,BSS,BBSS,Benini}. Thus, our present results
provide additional motivation and justification to exclude isotony from the list of 
axioms of locally covariant quantum field theory \cite{Brunetti}.
\end{rem}

\section{\label{sec:homotopyKan}Homotopy Kan extension}
The goal of this section is to construct toy-models of homotopical quantum 
field theories by using a homotopical generalization of the right Kan extension
\cite{hoKan1,hoKan2,hoKan3}. Instead of ordinary observable algebras,
these theories assign higher algebraic structures to spacetimes, which we shall model concretely
by {\em differential graded algebras}. As a consequence,
homotopical quantum field theories are more flexible and in particular they 
are able to capture detailed aspects of gauge theories that become invisible 
at the level of gauge invariant observables \cite{BSShomotopy}. 
By Remark \ref{rem:physicalinterpretation}, gauge symmetries also 
play an important role in our present work because one may loosely think of $\U\AA(M)$ 
as an algebra of gauge invariant observables, where the gauge symmetries
are the morphisms in $\pi^{-1}(M)$.
\sk

Let $\pi : \Str\to \Loc $ be a category fibered in groupoids over $\Loc$
and $\AA : \Str\to \Alg$ a functor. 
In practice, $\AA$ will satisfy the quantum field theory axioms of Definition \ref{def:QFT},
but these are not needed for the present section.
We may regard $\AA$ as a functor $\AA : \Str\to \dgAlg$ 
with values in the model category of differential graded algebras
by regarding algebras as differential graded algebras concentrated in degree zero (with trivial differential).
For a brief introduction to differential graded algebras and their homotopy theory we 
refer the reader to Appendix \ref{app:dgAlg}.
The homotopy right Kan extension of $\AA : \Str\to \dgAlg$ along
$\pi: \Str\to \Loc$ provides us with a functor 
\begin{flalign}
\hoRanpi\AA : \Loc\longrightarrow\dgAlg~
\end{flalign}
that may be computed ``pointwise''
by homotopy limits, see \cite{hoKan1,hoKan2,hoKan3}. Explicitly,
to any object $M$ in $\Loc$ this functor assigns the differential graded algebra
given by the homotopy limit
\begin{flalign}
\hoRanpi\AA(M) := \holim_{\dgAlg}^{}\Big(M\downarrow \pi \stackrel{\QQ^M}{\longrightarrow} \Str \stackrel{\AA}{\longrightarrow} \dgAlg\Big)
\end{flalign}
in the category $\dgAlg$. Using the explicit description of $\holim_{\dgAlg}^{}$ presented in
Appendix \ref{app:dgAlg}, we obtain that the graded vector space underlying $\hoRanpi\AA(M)$ is
\begin{subequations}\label{eqn:hoKanobject}
\begin{flalign}
\big(\hoRanpi\AA(M)\big)^0 &= \prod_{(S,h)\in (M\downarrow\pi)_0} \AA(S)~,\\
\big(\hoRanpi\AA(M)\big)^n &= \prod_{\mycom{(g_1,\dots,g_n)\in (M\downarrow\pi)_n}{g_i\neq\id} } \AA\big(\QQ^M(\mathrm{t}(g_1))\big)~,
\end{flalign}
\end{subequations}
for all $n\in\bbZ_{\geq 1}$, where $(M\downarrow\pi)_n$ denotes the degree $n$ component
of the nerve of the category  $M\downarrow\pi$, i.e.\
elements $(g_1,\dots,g_n)\in (M\downarrow\pi)_n$ are composable $n$-arrows
in $M\downarrow \pi$, and $\mathrm{t}(g)$ denotes the target of the $M\downarrow \pi$-morphism $g$. 
It is convenient to regard elements $a\in \hoRanpi\AA(M)^0$ as mappings
\begin{flalign}
(M\downarrow\pi)_0 \ni (S,h) \longmapsto a(S,h)\in \AA(S)
\end{flalign}
and elements $a\in \hoRanpi\AA(M)^n$, for $n\geq 1$, as mappings
\begin{subequations}
\begin{flalign}
(M\downarrow\pi)_n \ni (g_1,\dots,g_n) \longmapsto a(g_1,\dots,g_n)\in \AA\big(\QQ^M(\mathrm{t}(g_1))\big)
\end{flalign}
on all of $(M\downarrow\pi)_n$, which satisfy the normalization condition
\begin{flalign}
a(g_1,\dots,g_{i-1},\id,g_{i+1},\dots g_n)=0~,
\end{flalign}
\end{subequations}
for all $i=1,\dots,n$.
Notice that, compared to the general procedure to compute homotopy limits in $\dgAlg$ 
(cf.\ Appendix \ref{app:dgAlg}), some major simplifications occur in the present situation 
because the functor $\AA : \Str\to \dgAlg$ assigns differential graded algebras that are 
concentrated in degree $0$. As a consequence, the horizontal part of the double 
cochain complex \eqref{eqn:doublecomplex} is trivial. 
This is reflected also by the definitions of differential and product displayed below. 
The differential $\dd : \hoRanpi\AA(M)^n\to \hoRanpi\AA(M)^{n+1}$ is given by
\begin{multline}
\dd a(g_1,\dots,g_{n+1}) =  \AA(g_1)\big(a(g_2,\dots,g_{n+1})\big) \\
+ \sum_{i=1}^n (-1)^i\, a(g_1,\dots,g_i\circ g_{i+1},\dots,g_{n+1})
+ (-1)^{n+1}\, a(g_1,\dots,g_n)~,
\end{multline}
for all $a\in \hoRanpi\AA(M)^n$. The product on $\hoRanpi\AA(M)$ reads as
\begin{flalign}
(a\,a^\prime)(g_1,\dots,g_{n+n^\prime}) = a(g_1,\dots,g_n)\, \AA(g_1\circ\cdots\circ g_n)\big(a^\prime(g_{n+1},\dots, g_{n+n^\prime})\big)~,
\end{flalign}
for all $a\in \hoRanpi\AA(M)^n $ and $a^\prime\in \hoRanpi\AA(M)^{n^\prime}$,
and the unit element is given by
\begin{flalign}
\1(S,h) = \1 \in \AA(S)~.
\end{flalign}
\begin{rem}\label{rem:hoRantoRan}
Notice that the zeroth cohomology of the differential graded algebra $\hoRanpi\AA(M)$
given by \eqref{eqn:hoKanobject} is the algebra $\Ranpi\AA(M)$ that is assigned by the
ordinary right Kan extension, cf.\ \eqref{eqn:RanpiLimitexplicit}.
In fact, an element $a\in \hoRanpi\AA(M)^{0}$ is
specified by an {\em arbitrary} sequence of elements $a(S,h)\in \AA(S)$, for all objects $(S,h)$ in $M\downarrow \pi$,
and $\dd a =0$ is equivalent to
\begin{flalign}
\dd a(g) = \AA(g)\big(a(S,h)\big) - a(\widetilde{S},\widetilde{h}) =0~,
\end{flalign}
for all $M\downarrow \pi$-morphisms $g: (S,h)\to (\widetilde{S},\widetilde{h})$,
which is precisely the compatibility condition in \eqref{eqn:RanpiLimitexplicit}.
It thus follows that
\begin{flalign}
H^0\big(\hoRanpi\AA(M)\big) = \Ranpi\AA(M)~,
\end{flalign}
for each object $M$ in $\Loc$.
\end{rem}

It remains to explain how the functor $\hoRanpi\AA: \Loc\to \dgAlg$ 
acts on morphisms: Given any $\Loc$-morphism $f: M\to M^\prime$,
the $\dgAlg$-morphism $\hoRanpi\AA(f) : \hoRanpi\AA(M) \to 
\hoRanpi\AA(M^\prime)$ is specified in degree zero by
\begin{subequations}
\begin{flalign}
\big(\hoRanpi\AA(f)(a)\big)(S^\prime,h^\prime) = a(S^\prime,h^\prime\circ f)~,
\end{flalign}
for all $a\in \hoRanpi\AA(M)^0$ and $(S^\prime,h^\prime)\in (M^\prime\downarrow \pi)_0$,
and in degree $n\in \bbZ_{\geq 1}$ by
\begin{flalign}
\big(\hoRanpi\AA(f)(a)\big)(g_1^\prime,\dots,g_n^{\prime})= a\big(f^\ast(g_1^\prime,\dots,g_n^{\prime})\big)~,
\end{flalign}
\end{subequations}
for all $a\in \hoRanpi\AA(M)^n$ and composable $n$-arrows 
\begin{subequations}
\begin{flalign}
\xymatrix{
(S_0^\prime,h_0^\prime) & \ar[l]_-{g^\prime_1} (S_1^\prime,h_1^\prime) &\ar[l]_-{g^\prime_2} ~~\cdots ~~& \ar[l]_-{g^\prime_n}(S_n^\prime,h_n^\prime)
}
\end{flalign}
 in $M^\prime\downarrow \pi$,
where the composable $n$-arrow $f^\ast(g_1^\prime,\dots,g_n^{\prime})$ in $M\downarrow \pi$ explicitly reads as
\begin{flalign}
\xymatrix{
(S_0^\prime,h_0^\prime\circ f) & \ar[l]_-{g^\prime_1} (S_1^\prime,h_1^\prime\circ f) &\ar[l]_-{g^\prime_2} ~~\cdots ~~& \ar[l]_-{g^\prime_n}(S_n^\prime,h_n^\prime\circ f)
}~.
\end{flalign}
\end{subequations}
It is easy to check that $\hoRanpi\AA(f) : \hoRanpi\AA(M) \to \hoRanpi\AA(M^\prime)$ preserves
the differentials, products and units. Using Remark \ref{rem:hoRantoRan}, we further obtain
\begin{propo}
The composition of the homotopy right Kan extension $\hoRanpi\AA : \Loc\to\dgAlg$
with the zeroth cohomology functor $H^0 : \dgAlg\to\Alg$ yields the ordinary right Kan extension, i.e.\
$\Ranpi \AA= H^0\circ \hoRanpi\AA : \Loc \to \Alg$.
\end{propo}

We now shall prove a generalization of Theorem \ref{theo:initial},
which allows us describe (up to weak equivalence) the differential graded algebras $\hoRanpi\AA(M)$
arising from the homotopy right Kan extension 
in terms of the homotopy limit $\holim_{\dgAlg}^{}\AA\vert_{\pi^{-1}(M)}$ 
of the restricted functor $\AA\vert_{\pi^{-1}(M)} : \pi^{-1}(M)\to\dgAlg$.
Similarly to the non-homotopic case (cf.\ Section \ref{sec:properties}),
this reformulation will later be used in order to simplify the study of
properties of the homotopy right Kan extension.
\sk

Let us start with working out explicitly the homotopy limit $\holim_{\dgAlg}^{}\AA\vert_{\pi^{-1}(M)}$.
Using again Appendix \ref{app:dgAlg}, we obtain that the graded vector space 
underlying $\holim_{\dgAlg}^{}\AA\vert_{\pi^{-1}(M)}$ is
\begin{subequations}\label{eqn:holimobject}
\begin{flalign}
\big(\holim_{\dgAlg}^{}\AA\vert_{\pi^{-1}(M)}\big)^0 &= \prod_{S\in \pi^{-1}(M)_0} \AA(S)~,\\
\big(\holim_{\dgAlg}^{}\AA\vert_{\pi^{-1}(M)}\big)^n &= \prod_{\mycom{(g_1,\dots,g_n)\in \pi^{-1}(M)_n}{g_i\neq\id} } \AA\big(\mathrm{t}(g_1)\big)~.
\end{flalign}
\end{subequations}
Again, it is convenient to regard elements $a\in (\holim_{\dgAlg}^{}\AA\vert_{\pi^{-1}(M)})^0$ as mappings
\begin{flalign}
\pi^{-1}(M)_0 \ni S \longmapsto a(S)\in \AA(S)
\end{flalign}
and elements $a\in (\holim_{\dgAlg}^{}\AA\vert_{\pi^{-1}(M)})^n$, for $n\geq 1$, as mappings
\begin{subequations}
\begin{flalign}
\pi^{-1}(M)_n \ni (g_1,\dots,g_n) \longmapsto a(g_1,\dots,g_n)\in \AA\big(\mathrm{t}(g_1)\big)
\end{flalign}
on all of $\pi^{-1}(M)_n$, which satisfy the normalization condition
\begin{flalign}
a(g_1,\dots,g_{i-1},\id,g_{i+1},\dots g_n)=0~,
\end{flalign}
\end{subequations}
for all $i=1,\dots,n$.
The differential $\dd : (\holim_{\dgAlg}^{}\AA\vert_{\pi^{-1}(M)})^n\to 
(\holim_{\dgAlg}^{}\AA\vert_{\pi^{-1}(M)})^{n+1}$ is given by
\begin{multline}
\dd a(g_1,\dots,g_{n+1}) =  \AA(g_1)\big(a(g_2,\dots,g_{n+1})\big) \\
+ \sum_{i=1}^n (-1)^i\, a(g_1,\dots,g_i\circ g_{i+1},\dots,g_{n+1})
+ (-1)^{n+1}\, a(g_1,\dots,g_n)~,
\end{multline}
for all $a\in (\holim_{\dgAlg}^{}\AA\vert_{\pi^{-1}(M)})^n$. The product on $\holim_{\dgAlg}^{}\AA\vert_{\pi^{-1}(M)}$ 
reads as
\begin{flalign}\label{eqn:holimproduct}
(a\,a^\prime)(g_1,\dots,g_{n+n^\prime}) = a(g_1,\dots,g_n)\, \AA(g_1\circ\cdots\circ g_n)\big(a^\prime(g_{n+1},\dots, g_{n+n^\prime})\big)~,
\end{flalign}
for all $a\in (\holim_{\dgAlg}^{}\AA\vert_{\pi^{-1}(M)})^n$ and 
$a^\prime\in (\holim_{\dgAlg}^{}\AA\vert_{\pi^{-1}(M)})^{n^\prime}$,
and the unit element is given by
\begin{flalign}
\1(S) = \1 \in \AA(S)~.
\end{flalign}
The canonical mapping 
\begin{subequations}\label{eqn:hoRanpitoholimmapping}
\begin{flalign}
\kappa_M^{} : \hoRanpi\AA(M)\longrightarrow \holim_{\dgAlg}^{}\AA\vert_{\pi^{-1}(M)}
\end{flalign}
specified by
\begin{flalign}
\big(\kappa_M^0(a)\big)(S) := a(S,\id_M)~,
\end{flalign}
for all $a\in \hoRanpi\AA(M)^0$ and $S\in \pi^{-1}(M)_0$, and
\begin{flalign}\label{eqn:kappantmp}
\big(\kappa_M^n(a)\big)(g_1,\dots,g_n):= a(g_1,\dots,g_n)~,
\end{flalign}
\end{subequations}
for all $n\in\bbZ_{\geq 1}$, $a\in \hoRanpi\AA(M)^n$ and $(g_1,\dots,g_n)\in \pi^{-1}(M)_n$,
is a $\dgAlg$-morphism.
\begin{theo}\label{theo:homotopyinitial}
Let $\pi : \Str\to\Loc$ be a category fibered in groupoids (or just a fibered category) and $\AA :\Str\to\Alg$ a functor.
Then, for each object $M$ in $\Loc$, the $\dgAlg$-morphism \eqref{eqn:hoRanpitoholimmapping} is a 
weak equivalence in the model category $\dgAlg$.
\end{theo}
\begin{proof}
By \cite[Theorem 19.6.7]{Hirschhorn}, it is sufficient to prove that our functor $\iota : \pi^{-1}(M)\to M\downarrow \pi$
(cf.\ proof of Theorem \ref{theo:initial}) is homotopy initial. According to \cite[Definition 19.6.1]{Hirschhorn},
the functor  $\iota : \pi^{-1}(M)\to M\downarrow \pi$ is homotopy initial if for every object
$(S,h)$ in $M\downarrow \pi$ the nerve of its over-category $\iota\downarrow (S,h)$ is contractible as a simplicial set.
We shall now show that the category $\iota\downarrow (S,h)$ has a terminal object, which by
\cite[Proposition 14.3.14]{Hirschhorn} implies that its nerve is contractible and hence completes the proof.
\sk

Recall that objects in $\iota\downarrow (S,h)$ are pairs $(S^\prime, g^\prime)$, where $S^\prime$ is an object
in $\pi^{-1}(M)$ and $g^\prime : \iota(S^\prime) = (S^\prime,\id_M) \to (S,h)$ is an $M\downarrow\pi$-morphism.
We may visualize objects in $\iota\downarrow (S,h)$ by morphisms of the form
\begin{flalign}
\xymatrix{
(S^\prime,\id_M) \ar[r]^-{g^\prime} &(S,h)
}
\end{flalign}
in $M\downarrow\pi$.
Using that $\iota : \pi^{-1}(M)\to M\downarrow \pi$ is fully faithful,
a morphism $g : (S^\prime,g^\prime)\to (S^{\prime\prime},g^{\prime\prime})$ in  $\iota\downarrow (S,h)$
is given by a commutative triangle
\begin{flalign}
\xymatrix{
\ar[dr]_-{g^\prime }(S^\prime,\id_M)  \ar[rr]^-{g}&& (S^{\prime\prime},\id_M) \ar[dl]^-{g^{\prime\prime}}\\
&(S,h)&
}
\end{flalign}
in $M\downarrow\pi$. By Definition \ref{def:fiberedcategory}, there exists a pullback of $S$ to $M$
along the $\Loc$-morphism $h:M\to\pi(S)$ and we make an arbitrary choice 
$h_\ast :  h^\ast S \to S $ of such cartesian $\Str$-morphism. By definition,
we have that $\pi(h_\ast ) = h : M\to \pi(S)$ and hence we obtain an $M\downarrow\pi$-morphism
of the form $h_\ast : (h^\ast S,\id_M)\to (S,h)$ which defines an object in $\iota\downarrow (S,h)$.
Our claim is that this object is a terminal object in $\iota\downarrow (S,h)$. To prove this claim, we have to show that 
given any other object $(S^\prime,g^\prime)$ in $\iota\downarrow (S,h)$, there exists a unique
way to complete the diagram
\begin{flalign}
\xymatrix{
\ar[dr]_-{g^\prime }(S^\prime,\id_M)  \ar@{-->}[rr]^{\exists ! g}&& (h^\ast S ,\id_M) \ar[dl]^-{h_\ast}\\
&(S,h)&
}
\end{flalign}
in $M\downarrow\pi$. Using that $\pi(g^\prime) = \pi(h_\ast) =h : M\to\pi(S)$,
this is equivalent to completing the $\Str$-diagram
\begin{flalign}
\xymatrix{
\ar[dr]_-{g^\prime }S^\prime  \ar@{-->}[rr]^{\exists ! g}&& h^\ast S \ar[dl]^-{h_\ast}\\
&S&
}
\end{flalign}
by a unique $\Str$-morphism $g: S^\prime \to h^\ast S $ satisfying $\pi(g) =\id_M:M\to M$. 
Since $h_\ast : h^\ast S\to S$ is by construction a cartesian $\Str$-morphism,
existence and uniqueness of the sought $g$ are ensured, see Definition \ref{def:cartesian}. 
This shows that $h_\ast : (h^\ast S,\id_M)\to (S,h)$ 
is a terminal object in $\iota\downarrow (S,h)$ and completes the proof.
\end{proof}

We will now show that the weak equivalences \eqref{eqn:hoRanpitoholimmapping} 
may be inverted up to cochain homotopies.
Let us fix a cleavage on $\pi : \Str\to \Loc$, i.e.\ for each object $S$ in $\Str$ and 
each $\Loc$-morphism $f : M\to \pi(S)$ we make a choice of cartesian $\Str$-morphism
$f_\ast : f^\ast S \to S$ satisfying $\pi(f_\ast) = f: M\to \pi(S)$. 
In order to simplify some of the formulas below, we will choose all
${\id_M}_\ast : \id_M^\ast S \to S$ to be the identity $\Str$-morphisms
$\id_{S} : S\to S$. Given a choice of cleavage, we can assign to each element
$(g_1,\dots,g_n)\in (M\downarrow \pi)_n$, with $n\geq 1$,
an $n$-arrow $(g_1^h,\dots,g_n^h) \in \pi^{-1}(M)_n$
via the commutative diagram
\begin{flalign}\label{eqn:pullbackmorph}
\xymatrix{
(S_0,h_0) & \ar[l]_-{g_1} (S_1,h_1) & \ar[l]_-{g_2}~~ \cdots ~~ & \ar[l]_-{g_n} (S_n,h_n)\\
\ar[u]^-{{h_0}_\ast} (h_0^\ast S_0,\id_M) & \ar[u]_-{{h_1}_\ast}\ar[l]^-{g^h_1} (h_1^\ast S_1,\id_M) & \ar[l]^-{g^h_2}~~ \cdots ~~ & \ar[u]_-{{h_n}_\ast} \ar[l]^-{g^h_n} (h_n^\ast S_n,\id_M)
}
\end{flalign}
in $M\downarrow \pi$.
With these preparations, we define a $\dgAlg$-morphism
\begin{subequations}\label{eqn:holimtohoRanpimapping}
\begin{flalign}
\zeta_M^{} : \holim_{\dgAlg}^{} \AA\vert_{\pi^{-1}(M)}^{} \longrightarrow \hoRanpi\AA(M)
\end{flalign}
by setting
\begin{flalign}
\big(\zeta_M^0(a)\big)(S,h) := \AA(h_\ast)\big(a(h^\ast S)\big)~,
\end{flalign}
for all $a\in (\holim_{\dgAlg}^{} \AA\vert_{\pi^{-1}(M)})^0$ and $(S,h)\in (M\downarrow\pi)_0$,
and
\begin{flalign}
\big(\zeta_M^n(a)\big)(g_1,\dots,g_n) := \AA({h_0}_\ast)\big(a(g_1^h,\dots,g_n^h)\big)~,
\end{flalign}
\end{subequations}
for all $n \in \bbZ_{\geq 1}$, $a\in (\holim_{\dgAlg}^{} \AA\vert_{\pi^{-1}(M)})^n$
and $(g_1,\dots,g_n)\in (M\downarrow \pi)_n$.
Recalling the definition of $\kappa_M^{} : \hoRanpi\AA(M) \to \holim_{\dgAlg}^{}\AA\vert_{\pi^{-1}(M)}$, 
see \eqref{eqn:hoRanpitoholimmapping}, one easily confirms that
\begin{flalign}\label{eqn:kappacirczeta}
\kappa_M^{}\circ \zeta_M^{} = \id_{\holim_{\dgAlg}\AA\vert_{\pi^{-1}(M)}}~.
\end{flalign}
Here it is essential to use that for our choice of cleavage the pullbacks
${\id_M}_\ast : \id_M^\ast S \to S$ along the identity morphisms $\id_M$ 
are the identities $\id_S$. For an arbitrary choice of cleavage, equation 
\eqref{eqn:kappacirczeta} just holds true up to cochain homotopy.
\sk

The other composition $\zeta_M^{}\circ\kappa_M^{}$ is just cochain homotopic 
to the identity, i.e.\
\begin{flalign}\label{eqn:zetakappahomotopy}
\zeta_M^{}\circ\kappa_M^{} - \id_{\hoRanpi\AA(M)} = \eta_M^{}\circ \dd + \dd \circ \eta_M^{}~.
\end{flalign}
The cochain homotopy
\begin{subequations}
\begin{flalign}
\eta_M^{} : \hoRanpi\AA(M)^{\bullet+1} \longrightarrow \hoRanpi\AA(M)^\bullet
\end{flalign}
explicitly reads as
\begin{flalign}
\big(\eta_M^1(a)\big)(S,h) := a(h_\ast)~,
\end{flalign}
for all $a\in \hoRanpi\AA(M)^{1}$ and $(S,h)\in (M\downarrow \pi)_0$,
and
\begin{flalign}
\big(\eta_M^{n+1}(a)\big)(g_1,\dots,g_n) := \sum_{i=0}^{n} (-1)^i\, a(g_1,\dots,g_i, {h_{i}}_\ast, g_{i+1}^h,\dots,g_{n}^{h})~,
\end{flalign}
\end{subequations}
for all $n\in\bbZ_{\geq 1}$, $a\in \hoRanpi\AA(M)^{n+1}$ and $(g_1,\dots,g_n)\in (M\downarrow \pi)_n$.
The verification of \eqref{eqn:zetakappahomotopy} is a straightforward, but slightly lengthy, computation. 
We suggest the reader to explore the pattern in low degree before attacking the full calculation. 
\sk

Similarly to \eqref{eqn:UpiAAfunctor}, we may now define a more convenient and 
efficient model for the homotopy right Kan extension by using the weak equivalences
$\kappa_M^{}$ \eqref{eqn:hoRanpitoholimmapping} and their inverses (up to homotopy) 
$\zeta_M^{}$ \eqref{eqn:holimtohoRanpimapping}. 
In contrast to \eqref{eqn:UpiAAfunctor}, our present 
construction does {\em not} equip the assignment of differential 
graded algebras $M\mapsto \holim_{\dgAlg}^{} \AA\vert_{\pi^{-1}(M)}^{} $
with a strict functorial structure, but only with a functorial structure
`up to homotopy'. The reason for this is that $\kappa_M^{}$ and $\zeta_M^{}$
are inverse to each other only up to homotopy.
Concretely, the construction is as follows: 
We define the assignment 
\begin{subequations}\label{eqn:hoUpiAAfunctor}
\begin{flalign}
\hoU \AA : \Loc\longrightarrow \dgAlg
\end{flalign}
by setting
\begin{flalign}
\hoU\AA(M) := \holim_{\dgAlg}^{} \AA\vert_{\pi^{-1}(M)}^{}~,
\end{flalign}
for all objects $M$ in $\Loc$, and
\begin{flalign}
\hoU\AA(f) := \kappa_{M^\prime}^{} \circ \hoRanpi\AA(f) \circ \zeta_M^{} : \hoU\AA(M)\longrightarrow \hoU\AA(M^\prime)~,
\end{flalign}
\end{subequations}
for all $\Loc$-morphisms $f:M\to M^\prime$. Explicitly, the $\dgAlg$-morphism $\hoU\AA(f)$
acts in degree $0$ as
\begin{subequations}\label{eqn:hoUpiAAmorphexplicit}
\begin{flalign}
\big(\hoU\AA(f)(a)\big)(S^\prime) = \AA(f_\ast)\big(a(f^\ast S^\prime)\big)~,
\end{flalign}
for all $a\in \hoU\AA(M)^0$ and $S^\prime\in\pi^{-1}(M^\prime)_0$,
and in degree $n\geq 1$ as
\begin{flalign}
\big(\hoU\AA(f)(a)\big)(g_1^\prime,\dots,g_n^\prime) = \AA(f_\ast)\big(a(g_1^{\prime\,f},\dots,g_n^{\prime\,f})\big)~,
\end{flalign}
for all $a\in \hoU\AA(M)^n$ and $(g_1^\prime,\dots,g_n^\prime)\in\pi^{-1}(M^\prime)_n$,
where similarly to \eqref{eqn:pullbackmorph} the $n$-arrow $(g_1^{\prime\,f},\dots,g_n^{\prime\,f})\in\pi^{-1}(M)_n$
is defined by pullback along $f$ of the $n$-arrow $(g_1^\prime,\dots,g_n^\prime)\in\pi^{-1}(M^\prime)_n$. Concretely, 
$(g_1^{\prime\,f},\dots,g_n^{\prime\,f})\in\pi^{-1}(M)_n$ is defined by the commutative diagram
\begin{flalign}\label{eqn:pullbackmorphallsame}
\xymatrix{
(S^\prime_0 ,f) & \ar[l]_-{g^\prime_1} (S^\prime_1,f) & \ar[l]_-{g^\prime_2}~~ \cdots ~~ & \ar[l]_-{g^\prime_n} (S^\prime_n,f)\\
\ar[u]^-{{f}_\ast} (f^\ast S^\prime_0,\id_M) & \ar[u]_-{{f}_\ast}\ar[l]^-{g^{\prime\,f}_1} (f^\ast S^\prime_1,\id_M) & \ar[l]^-{g^{\prime f}_2}~~ \cdots ~~ & \ar[u]_-{{f}_\ast} \ar[l]^-{g^{\prime\,f}_n} (f^\ast S^\prime_n,\id_M)
}
\end{flalign}
\end{subequations}
in  $M\downarrow \pi$.
Due to our special choice of cleavage, 
\eqref{eqn:hoUpiAAfunctor} preserves identities, i.e.\
\begin{flalign}
\hoU\AA(\id_M) = \kappa_M^{}\circ \hoRanpi\AA(\id_M)\circ \zeta_M^{}= \kappa_M^{}\circ \zeta_M^{} =\id_{\hoU\AA(M)}~,
\end{flalign}
for all objects $M$ in $\Loc$. However, compositions are only preserved up to homotopy, i.e.\
there exists a cochain homotopy $\hoU\AA(f^\prime) \circ \hoU\AA(f) \sim \hoU\AA(f^\prime\circ f)$. 
Indeed, using \eqref{eqn:zetakappahomotopy}, one finds that
\begin{subequations}\label{eqn:compositionhomotopies}
\begin{flalign}
\hoU\AA(f^\prime) \circ \hoU\AA(f) = \hoU\AA(f^\prime\circ f) 
+ \gamma_{f^\prime,f}^{}\circ \dd + \dd \circ \gamma_{f^\prime,f}^{}~,
\end{flalign}
for all composable $\Loc$-morphisms 
$f : M\to M^\prime$ and $f^\prime: M^\prime\to M^{\prime\prime}$,
where 
\begin{flalign}
\gamma_{f^\prime,f}^{} := \kappa_{M^{\prime\prime}}^{}\circ \hoRanpi\AA(f^\prime) 
\circ \eta_{M^\prime}^{} \circ \hoRanpi\AA(f)\circ \zeta_{M}^{}~.
\end{flalign}
\end{subequations}
In the following we shall use the `up to homotopy' functor \eqref{eqn:hoUpiAAfunctor}
as a model for the homotopy right Kan extension.
\begin{rem}\label{rem:homphysicalinterpretation}
The differential graded algebra $\hoU\AA(M)$ assigned to a spacetime $M$
by our model \eqref{eqn:hoUpiAAfunctor} for the homotopy
right Kan extension has an interpretation in terms of
groupoid cohomology: From the explicit description 
of the homotopy limit, see below \eqref{eqn:holimobject}, 
we observe that $\hoU\AA(M) = C^\bullet(\pi^{-1}(M);\AA)$
is the differential graded algebra underlying the groupoid cohomology
of $\pi^{-1}(M)$ with values in the functor $\AA: \Str\to\Alg$. 
(See e.g.\ \cite{Crainic}  for some background material on groupoid cohomology.) 
Taking cohomologies, we obtain a graded algebra 
$H^\bullet (\hoU\AA(M)) = H^\bullet(\pi^{-1}(M);\AA)$, whose zeroth degree
is the algebra $\U\AA(M)$ assigned by the ordinary right Kan extension 
\eqref{eqn:UpiAAfunctor} to the spacetime $M$.
This observation suggests that the information about the action of gauge transformations on the
quantum field theory $\AA :\Str\to\Alg$ that is captured by the homotopy right Kan extension 
is more detailed than the one available in the ordinary right Kan extension $\U\AA :\Loc\to \Alg$.
The extra information is encoded in the higher cohomologies
$H^n (\hoU\AA(M)) = H^n(\pi^{-1}(M);\AA)$, for $n\geq 1$. Unfortunately, 
the physical interpretation of such higher-order information is currently not fully clear to us.
\sk

As a side remark, notice that the assignment of cohomologies $M\mapsto H^\bullet (\hoU\AA(M))$
is a strict functor, even though $\hoU\AA$ is just a functor `up to homotopy'. 
In fact, by construction the cohomology functors send cochain homotopies to identities.
Observe, moreover, that in the special case of a ``global gauge group'' $G$ (cf.\ Example \ref{ex:globalgauge})
groupoid cohomology reduces to group cohomology, i.e.\ $\hoU\AA(M) = C^\bullet(G;\AA)$
and $H^\bullet (\hoU\AA(M)) = H^\bullet(G;\AA)$.
\end{rem}

\begin{rem}\label{rem:uptohomotopyfunctor}
From a homotopical perspective, it would be natural to refine 
our concept of `up to homotopy' functor $\hoU\AA :\Loc\to \dgAlg$ 
by adding coherence conditions: Instead of just demanding that 
there exists a cochain homotopy 
$\hoU\AA(f^\prime) \circ \hoU\AA(f) \sim \hoU\AA(f^\prime\circ f)$
controlling compositions, one should make a particular choice 
for every pair of composable morphisms $f$ and $f^\prime$ 
(e.g.\ $\gamma_{f^\prime,f}^{}$ given in \eqref{eqn:compositionhomotopies})
and add this choice to the data defining an `up to homotopy' functor.
This is however just the first step towards a homotopically coherent description:
Given three composable $\Loc$-morphisms $f: M\to M^\prime$, $f^\prime: M^\prime\to M^{\prime\prime}$
and $f^{\prime\prime} : M^{\prime\prime}\to M^{\prime\prime\prime}$, we may compare the 
two cochain homotopies corresponding to compositions in different orders, i.e.\
\begin{subequations}
\begin{flalign}
\hoU\AA(f^{\prime\prime})\circ \big(\hoU\AA(f^\prime)\circ \hoU\AA(f)\big)
\end{flalign}
and 
\begin{flalign}
\big(\hoU\AA(f^{\prime\prime})\circ \hoU\AA(f^\prime) \big) \circ \hoU\AA(f)~.
\end{flalign}
\end{subequations}
It turns out these these two cochain homotopies are homotopic by a higher cochain homotopy.
Explicitly,  the difference of the two cochain homotopies is given by
\begin{subequations}
\begin{flalign}\label{eqn:coherence}
\gamma_{f^{\prime\prime}, f^{\prime} \circ f}^{} + \hoU\AA(f^{\prime\prime})\circ \gamma_{f^\prime,f}^{}
-\big( \gamma_{f^{\prime\prime}\circ f^{\prime},f}^{} + \gamma_{f^{\prime\prime},f^{\prime}}^{}\circ \hoU\AA(f)\big)
= \dd\circ \gamma_{f^{\prime\prime},f^{\prime},f}^{} - \gamma_{f^{\prime\prime},f^{\prime},f}^{}\circ\dd~,
\end{flalign}
where
\begin{flalign}\label{eqn:gamma3}
\gamma_{f^{\prime\prime},f^{\prime},f}^{} := \kappa_{M^{\prime\prime\prime}}^{}\circ \hoRanpi\AA(f^{\prime\prime})\circ \eta_{M^{\prime\prime}}^{}\circ 
\hoRanpi\AA(f^\prime)\circ\eta_{M^{\prime}}^{} \circ\hoRanpi\AA(f)\circ \zeta_M^{}~.
\end{flalign}
\end{subequations}
A particular choice of such higher cochain homotopies 
(e.g.\ $\gamma_{f^{\prime\prime},f^{\prime},f}^{}$ given in \eqref{eqn:gamma3})
should be added to the data defining an `up to homotopy' functor. It is crucial to 
notice that the cochain homotopies $\gamma_{f^\prime,f}^{}$ and 
higher cochain homotopies $\gamma_{f^{\prime\prime},f^{\prime},f}^{}$ have to 
satisfy the coherence conditions \eqref{eqn:coherence}.
Considering compositions of four and more morphisms introduces
additional higher cochain homotopies and coherence conditions, which all should be added to the definition
of `up to homotopy' functor.
\sk

From the above description it becomes evident that adding coherent homotopies
to our definition of `up to homotopy' functor $\hoU\AA :\Loc\to \dgAlg$ is
a very cumbersome task if we restrict ourselves to
elementary categorical techniques.
The right framework to systematically address these issues lies in the theory of 
colored operads. (We are very grateful to Ulrich Bunke for suggesting this operadic picture to us.)
In this framework, coherent `up to homotopy' functors may be naturally defined
as homotopy coherent diagrams, which are algebras (in the operadic sense) over 
the cofibrant replacement of the diagram operad, see e.g.\ \cite{Operad}. 
By such operadic techniques, in particular the homotopy transfer theorem,
we can already infer that our `up to homotopy' functor $\hoU\AA :\Loc\to \dgAlg$
is a homotopy coherent diagram in the sense of \cite{Operad} because
its `up to homotopy' functoriality is transferred from the the strict functoriality of $\hoRanpi\AA :\Loc\to \dgAlg$ 
via the weak equivalences $\kappa_M^{}$ \eqref{eqn:hoRanpitoholimmapping} 
and $\zeta_M^{}$ \eqref{eqn:holimtohoRanpimapping}, see  \eqref{eqn:hoUpiAAfunctor}. 
\sk

In the next section, we shall observe that dealing with coherent
versions of commutativity (in the sense of the causality axiom) 
`up to homotopy' leads to similar technical issues as above. 
See in particular Remark \ref{rem:uptohomotopycommutativity} 
for further comments. This suggest the development of an operadic
framework for locally covariant quantum field theory and its 
homotopical generalization. In this way all higher-order coherences
would be automatically encoded in the framework. This is similar to the recent factorization algebra 
approach to quantum field theory by Costello and Gwilliam \cite{Costello}, 
however using a different colored operad that captures the causal 
structure of {\em Lorentzian} spacetime manifolds.
Developing such an operadic framework for locally covariant quantum field theory
is beyond the scope of the present paper, but we plan to come back to this
in future works. 
\end{rem}

\section{\label{sec:homproperties}Homotopical properties}
Let  $\AA : \Str \to \Alg$ be a quantum field theory on a category fibered in groupoids
$\pi : \Str\to \Loc$ in the sense of Definition \ref{def:QFT}. 
In the previous section we obtained
a convenient description of its homotopy right Kan extension in terms of an `up to homotopy'
functor $\hoU\AA :\Loc\to \dgAlg$ with values in the category of differential graded algebras, 
cf.\  \eqref{eqn:hoUpiAAfunctor}.
We will now address the question whether $\hoU\AA :\Loc\to \dgAlg$ is a {\em homotopical}
quantum field theory, i.e.\ whether it satisfies homotopically meaningful generalizations
of the axioms proposed in \cite{Brunetti}. We shall focus only on the causality and
the time-slice axiom because, as we have seen in Section \ref{sec:properties},
isotony is violated for almost all of our examples of interest.
\sk

We start with the causality axiom. Given a causally disjoint $\Loc$-diagram 
$M_1\stackrel{f_1}{\longrightarrow} M \stackrel{f_2}{\longleftarrow} M_2$,
we consider the induced $\bbZ_{\geq 0}$-graded commutator
\begin{flalign}
[\,\cdot\,,\,\cdot\,] \circ \big(\hoU\AA(f_1)\otimes \hoU\AA(f_2)\big) \,:\, \hoU\AA(M_1)\otimes \hoU\AA(M_2)
 \longrightarrow \hoU\AA(M)~,
\end{flalign}
which is a $\dgVec$-morphism. (Here $\otimes$ denotes the tensor product of differential graded vector spaces.)
Using the explicit expression \eqref{eqn:holimproduct} for the product on $\hoU\AA(M)$, we observe that
the graded commutator $[\hoU\AA(f_1)(a), \hoU\AA(f_2)(a^\prime) ]$ vanishes only if
both $a$ and $a^\prime$ are of degree $0$. However, if $a$ or $a^\prime$ (or both) are of degree $\geq 1$,
then the graded commutator in general does not vanish. Hence, 
our `up to homotopy' functor $\hoU\AA :\Loc\to \dgAlg$ does not 
satisfy the original form of the causality axiom.
\sk

The fact that $\hoU\AA :\Loc\to \dgAlg$ does not satisfy the original 
form of the causality axiom is not problematic at all (even more, it 
should be expected), because strict commutativity is not a homotopically 
meaningful concept, i.e.\ it is not preserved under weak equivalences in $\dgVec$.
The homotopically meaningful replacement of the causality axiom is
commutativity `up to homotopy', i.e.\ for each 
causally disjoint $\Loc$-diagram 
$M_1\stackrel{f_1}{\longrightarrow} M \stackrel{f_2}{\longleftarrow} M_2$
there should exist a cochain homotopy $[\,\cdot\,,\,\cdot\,]  \sim 0$ to the zero map.
Hence, in order to prove that $\hoU\AA :\Loc\to \dgAlg$ satisfies causality `up to homotopy', 
we must find 
\begin{flalign}
\lambda_{f_1,f_2}^{} : \big(\hoU\AA(M_1)\otimes \hoU\AA(M_2)\big)^{\bullet+1} \longrightarrow  \hoU\AA(M)^\bullet~,
\end{flalign}
such that
\begin{flalign}\label{eqn:uptohomotopycommutator}
[\,\cdot\,,\,\cdot\,] \circ \big(\hoU\AA(f_1)\otimes \hoU\AA(f_2)\big) = 
\lambda_{f_1,f_2}^{} \circ\dd + \dd\circ \lambda_{f_1,f_2}^{}~,
\end{flalign}
for each causally disjoint $\Loc$-diagram 
$M_1\stackrel{f_1}{\longrightarrow} M \stackrel{f_2}{\longleftarrow} M_2$.
Our goal is to construct such cochain homotopies and thereby to show
that  $\hoU\AA :\Loc\to \dgAlg$ satisfies the causality axiom `up to homotopy'.
\begin{rem}\label{rem:uptohomotopycommutativity}
Similarly to Remark \ref{rem:uptohomotopyfunctor}, it would be natural from a homotopical perspective
to refine this notation of causality `up to homotopy' by adding the (higher) homotopies 
and their coherence conditions to the data of a homotopical quantum field theory. 
These additional structures would capture the homotopical information encoded in the 
commutation of more than two observables arising from
families of mutually causally disjoint embeddings $f_i : M_i \to M$, with $i=1,2,\dots, N$.
For example, given $N=3$ 
and observables $a_i$ localized in $M_i$, these structures
relate the cochain homotopy for the 2-step 
commutation $a_1\,a_2\,a_3 \to a_1\,a_3\,a_2\to a_3\,a_1\,a_2$ (i.e.\ first commuting
$a_2$ with $a_3$ and then commuting $a_1$ with $a_3$)
and the cochain homotopy for the 1-step commutation $a_1\,a_2\,a_3 \to a_3\,a_1\,a_2$
(i.e.\ immediately commuting $a_1\,a_2$ with $a_3$) by a higher cochain homotopy.
An operadic point of view on locally covariant quantum field theory would be capable to
systematically address such a refined notion of causality `up to homotopy'
by using cofibrant replacements of colored operads \cite{Operad}.
Interestingly, the algebraic structures we are looking for resemble a colored operad
that interpolates between the $A_\infty$ and the $E_{\infty}$-operad,
depending on the causal relations between subspacetimes.
This will be studied and clarified in future works. 
\sk

As a side remark, notice that our present (non-coherent) notion of causality `up to homotopy' 
is sufficiently strong to imply that the (strictly functorial) assignment of cohomologies 
$M\mapsto H^\bullet(\hoU\AA(M))$ 
satisfies strict causality (in the sense of graded algebras). Hence, all information about the homotopical
quantum field theory $\hoU\AA(M)$ that is contained in its cohomologies behaves in a strictly causal way.
\end{rem}

Our method for establishing the cochain homotopies in 
\eqref{eqn:uptohomotopycommutator} is inspired by the treatment 
of the cup product in singular cohomology, see e.g.\ \cite[Proof of Theorem 3.11]{Hatcher}.
Let us first define a $\dgVec$-morphism
\begin{subequations}\label{eqn:reversemap}
\begin{flalign}
\rho_M^{} : \hoU\AA(M) \longrightarrow \hoU\AA(M)~,
\end{flalign}
which reverses the direction of an $n$-arrow in $\pi^{-1}(M)$ (notice that for this it is crucial
that $\pi^{-1}(M)$ is a groupoid). Explicitly, we set
\begin{flalign}
\big(\rho_M^0(a)\big)(S):= a(S)~,
\end{flalign}
for all $a\in \hoU\AA(M)^0$ and $S\in \pi^{-1}(M)_0$,
and
\begin{flalign}
\big(\rho_M^n(a)\big)(g_1,\dots,g_n):= 
(-1)^{\frac{n (n+1)}{2}}~\AA(g_1\circ\cdots\circ g_n)\big(a(g_n^{-1},\dots,g_1^{-1})\big)~,
\end{flalign}
\end{subequations}
for all $n\in\bbZ_{\geq 1}$, $a \in \hoU\AA(M)^n$ and $(g_1,\dots,g_n)\in \pi^{-1}(M)_n$.
The sign factor is motivated by the fact that $n(n+1)/2$ is the number of transpositions of
adjacent elements taking the string $(1,2,\dots,n)$ to the string $(n,n-1,\dots, 1)$.
Notice that reversing twice gives the identity, i.e.\ $\rho_M^{}\circ \rho_M^{} = \id_{\hoU\AA(M)}$.
A crucial property of $\rho_M^{}$ is that it is cochain homotopic to the identity $\id_{\hoU\AA(M)}$.
Let us define
\begin{subequations}\label{eqn:reversehomotopy}
\begin{flalign}
\beta_M^{} : \hoU\AA(M)^{\bullet+1} \longrightarrow \hoU\AA(M)^\bullet
\end{flalign}
by setting
\begin{flalign}
\big(\beta_M^{1}(a)\big)(S):=0~,
\end{flalign}
for all $a\in \hoU\AA(M)^1$ and $S\in\pi^{-1}(M)_0$, and
\begin{flalign}
\big(\beta_M^{n+1}(a)\big)(g_1,\dots,g_n) := (-1)^n \sum_{i=1}^n (-1)^{\frac{(n-i)(n-i+1)}{2}}\, a(g_1,\dots,g_{i-1}, g_{i}\circ\cdots\,\circ g_n,g_n^{-1},\dots,g_i^{-1} ) ~,
\end{flalign}
\end{subequations}
for all $n\in\bbZ_{\geq 1}$, $a\in \hoU\AA(M)^{n+1}$ and $(g_1,\dots,g_n)\in\pi^{-1}(M)_n$.
\begin{lem}\label{lem:reversehomotopy}
The equality
\begin{flalign}\label{eqn:reversehomotopyequality}
\rho_M^{} - \id_{\hoU\AA(M)} = \beta_M^{}\circ\dd + \dd\circ \beta_M^{}
\end{flalign}
holds true.
\end{lem}
\begin{proof}
In degree $n=0$ the equality holds true because $\rho_M^0 = \id_{\hoU\AA(M)^0}$ and $\beta_M^{1}=0$.
Degree $n=1$ already requires a short calculation: For all $a\in \hoU\AA(M)^1$ and $g\in \pi^{-1}(M)_1$,
\begin{flalign}
\nn \big((\beta_M^{2}\circ\dd + \dd\circ \beta_M^{1})(a)\big)(g) &= \big(\beta_M^2(\dd a)\big)(g) = -\dd a(g,g^{-1})\\
&= - \AA(g)\big(a(g^{-1})\big) + a(g\circ g^{-1}) - a(g)  = \big(\rho_M^{1}(a)\big)(g) - a(g)~,
\end{flalign}
where we also have used the normalization condition $a(\id) =0$. In degree $n\geq 2$, the equality
\eqref{eqn:reversehomotopyequality} can be proven with a straightforward but rather lengthy calculation using also the 
normalization conditions $a(g_1,\dots,g_{i-1},\id,g_{i+1},\dots,g_n)=0$.
As this calculation is not instructive, we shall not spell it out in detail.
\end{proof}

The role of $\rho_M^{}$ is to reverse the order of the product $\mu$ on $\hoU\AA(M)$ when evaluated on elements
associated to causally disjoint subsets in $M$.
\begin{lem}\label{lem:mumuop}
For any causally disjoint $\Loc$-diagram $M_1\stackrel{f_1}{\longrightarrow} M \stackrel{f_2}{\longleftarrow} M_2$,
the equality
\begin{flalign}\label{eqn:rhomuop}
\rho_M^{}\circ \mu\circ \big(\hoU\AA(f_1)\otimes \hoU\AA(f_2)\big) 
= \mu^\op \circ (\rho_M^{}\otimes\rho_M^{})\circ  \big(\hoU\AA(f_1)\otimes \hoU\AA(f_2)\big) 
\end{flalign}
holds true, where $\mu^\op$ is the opposite product on $\hoU\AA(M)$, i.e.\
$\mu^\op(a\otimes a^\prime) := a\cdot^\op a^\prime := (-1)^{n\, n^\prime}\, a^\prime \, a$, 
for all $a\in \hoU\AA(M)^n$ and $a^\prime \in \hoU\AA(M)^{n^\prime}$.
\end{lem}
\begin{proof}
Let $a\in\hoU\AA(M_1)^n$ and $a^\prime\in \hoU\AA(M_2)^{n^\prime}$ be arbitrary.
Using \eqref{eqn:hoUpiAAmorphexplicit}, we obtain for the left-hand side of
\eqref{eqn:rhomuop}
\begin{flalign}
\nn &\rho_M^{}\big(\hoU\AA(f_1)(a)~ \hoU\AA(f_2)(a^\prime)\big)(g_1,\dots,g_{n+n^\prime})\\
\nn &~\qquad~ =(-1)^{\frac{(n+n^\prime)(n+n^\prime+1)}{2}}~\AA(g_1\circ \cdots\circ g_{n+n^\prime})
\left(\big(\hoU\AA(f_1)(a)~ \hoU\AA(f_2)(a^\prime)\big)(g_{n+n^\prime}^{-1},\dots,g_1^{-1})\right)\\
\nn &~\qquad~ =(-1)^{\frac{(n+n^\prime)(n+n^\prime+1)}{2}}~\AA(g_1\circ \cdots\circ g_{n+n^\prime})
\big(\hoU\AA(f_1)(a)(g_{n+n^\prime}^{-1},\dots,g_{n^\prime+1}^{-1}) \big)\\
\nn &~\qquad~ \qquad~\qquad~\qquad~\qquad~\qquad~\qquad~\times \AA(g_1\circ \cdots\circ g_{n^\prime})
\big(\hoU\AA(f_2)(a^\prime)(g_{n^\prime}^{-1},\dots,g_{1}^{-1}) \big)\\
\nn &~\qquad~ =(-1)^{\frac{(n+n^\prime)(n+n^\prime+1)}{2}}~\AA(g_1\circ \cdots\circ g_{n+n^\prime}\circ {f_1}_\ast)
\big(a(g_{n+n^\prime}^{-1\,f_1},\dots,g_{n^\prime+1}^{-1\,f_1}) \big)\\
 &~\qquad~ \qquad~\qquad~\qquad~\qquad~\qquad~\qquad~\times \AA(g_1\circ \cdots\circ g_{n^\prime}\circ {f_{2}}_\ast)
\big(a^\prime(g_{n^\prime}^{-1\,f_2},\dots,g_{1}^{-1\,f_2}) \big)~.\label{eqn:tmpcausal}
\end{flalign}
For the right-hand side of \eqref{eqn:rhomuop} we obtain
\begin{flalign}
\nn &\big(\rho_M^{}(\hoU\AA(f_1)(a))\cdot^\op \rho_M^{}(\hoU\AA(f_2)(a^\prime))\big)(g_1,\dots,g_{n+n^\prime})\\
\nn &~\qquad~ =(-1)^{n\,n^\prime}~\big(\rho_M^{}(\hoU\AA(f_2)(a^\prime))~ \rho_M^{}(\hoU\AA(f_1)(a))\big)(g_1,\dots,g_{n+n^\prime})\\
\nn &~\qquad~ =(-1)^{n\,n^\prime}~\rho_M^{}(\hoU\AA(f_2)(a^\prime))(g_1,\dots,g_{n^\prime})\\
\nn &~\qquad~ \qquad~\qquad~\qquad~\qquad~\qquad~\qquad~\times \AA(g_1\circ\cdots\circ g_{n^\prime})\big( \rho_M^{}(\hoU\AA(f_1)(a))(g_{n^\prime+1},\dots,g_{n+n^\prime})\big)\\
\nn &~\qquad~ =(-1)^{n\,n^\prime}\,(-1)^{\frac{n^\prime(n^\prime + 1)}{2}}\,(-1)^{\frac{n(n+ 1)}{2}}~
\AA(g_1\circ \cdots\circ g_{n^\prime})\big(\hoU\AA(f_2)(a^\prime)(g_{n^\prime}^{-1},\dots,g_{1}^{-1})\big)\\
\nn &~\qquad~ \qquad~\qquad~\qquad~\qquad~\qquad~\qquad~\times \AA(g_1\circ\cdots\circ g_{n+n^\prime})\big( \hoU\AA(f_1)(a)(g_{n+n^\prime}^{-1},\dots,g_{n^\prime+1}^{-1})\big)\\
\nn &~\qquad~ = (-1)^{\frac{(n+n^\prime)(n+n^\prime + 1)}{2}}~
\AA(g_1\circ \cdots\circ g_{n^\prime}\circ {f_2}_\ast)\big(a^\prime(g_{n^\prime}^{-1\,f_2},\dots,g_{1}^{-1\,f_2})\big)\\
&~\qquad~ \qquad~\qquad~\qquad~\qquad~\qquad~\qquad~\times \AA(g_1\circ\cdots\circ g_{n+n^\prime}\circ{f_1}_\ast)\big( a(g_{n+n^\prime}^{-1\,f_{1}},\dots,g_{n^\prime+1}^{-1\,f_1})\big)~.\label{eqn:tmpcausal2}
\end{flalign}
Notice that the $\Str$-morphism
$g_1\circ\cdots\circ g_{n+n^\prime}\circ{f_1}_\ast : f_1^\ast S_{n+n^\prime}\to S_0$ projects down via $\pi$
to the $\Loc$-morphism $f_1: M_1\to M$ and that 
$g_1\circ \cdots\circ g_{n^\prime}\circ {f_2}_\ast : f_2^\ast S_{n^\prime}\to S_0$ projects down to
$f_2: M_2\to M$. By hypothesis, $f_1$ and $f_2$ are causally disjoint and $\AA : \Str\to \Alg$ satisfies the causality axiom, 
hence we can commute the two factors in the last step of \eqref{eqn:tmpcausal2} 
and thereby show that \eqref{eqn:tmpcausal} is equal to \eqref{eqn:tmpcausal2}. 
\end{proof}

With these preparations we can now verify the `up to homotopy' causality axiom.
\begin{theo}\label{theo:homotopicalcausality}
Let $\pi: \Str\to \Loc$ be a category fibered in groupoids
and $\AA :\Str \to\Alg$ a quantum field theory in the sense of Definition \ref{def:QFT}.
Then the homotopy right Kan extension $\hoU\AA : \Loc\to \dgAlg$ (cf.\ \eqref{eqn:hoUpiAAfunctor}) 
satisfies the causality axiom `up to homotopy'. Explicitly, given any causally disjoint $\Loc$-diagram
$M_1\stackrel{f_1}{\longrightarrow} M \stackrel{f_2}{\longleftarrow} M_2$, 
a cochain homotopy between the induced 
$\bbZ_{\geq 0}$-graded commutator and zero \eqref{eqn:uptohomotopycommutator}  is given by
\begin{flalign}
\lambda_{f_1,f_2}^{}  := \Big(\mu^\op\circ  \big(\rho_M^{}\otimes\beta_M^{}+ \beta_M^{}\otimes\id_{\hoU\AA(M)}\big) -\beta_{M}^{}\circ \mu \Big)\circ \big(\hoU\AA(f_1)\otimes\hoU\AA(f_2)\big)~,
\end{flalign}
where $\beta_M^{}$ is defined in \eqref{eqn:reversehomotopy} and $\mu^{(\op)}$ is the (opposite) 
product on $\hoU\AA(M)$.
\end{theo}
\begin{proof}
Notice that the $\bbZ_{\geq 0}$-graded commutator 
is the difference between the product and the opposite product,
i.e.\ $[\,\cdot\,,\,\cdot\,] = \mu - \mu^\op$. By Lemma \ref{lem:reversehomotopy}, we can
replace $\mu$ in this expression by
\begin{flalign}
\mu = \rho_M^{}\circ \mu - \beta_M^{}\circ \mu\circ\dd - \dd\circ \beta_M^{}\circ \mu~.
\end{flalign}
Composing $[\,\cdot\,,\,\cdot\,] $ from the right with $L := \hoU\AA(f_1)\otimes\hoU\AA(f_2)$
and using Lemma \ref{lem:mumuop}, we obtain
\begin{flalign}
\nn [\,\cdot\,,\,\cdot\,] \circ L
&=  \big(\rho_M^{}\circ \mu - \mu^\op \big)\circ L- \beta_M^{}\circ\mu\circ L \circ \dd 
- \dd\circ \beta_M^{}\circ \mu\circ L\\
&= \mu^\op\circ \big(\rho_M^{}\otimes\rho_M^{}-\id\otimes\id\big)\circ L - \beta_M^{}\circ\mu\circ L \circ \dd 
- \dd\circ \beta_M^{}\circ \mu\circ L~.\label{eqn:tmpcausalhom}
\end{flalign}
Using Lemma \ref{lem:reversehomotopy}, we obtain that 
$\rho_M^{}\otimes\rho_M^{}-\id\otimes\id$ is cochain homotopic to zero, i.e.\
\begin{flalign}
\nn \rho_M^{}\otimes\rho_M^{}-\id\otimes\id  &= \rho_M^{}\otimes(\rho_M^{}-\id) + (\rho_M^{}-\id)\otimes\id\\
\nn &=(\rho_M^{}\otimes\beta_M^{}) \circ\dd + \dd\circ (\rho_M^{}\otimes\beta_M^{}) + (\beta_M^{}\otimes\id)\circ\dd
+ \dd\circ (\beta_M^{}\otimes\id)\\
&=\big(\rho_M^{}\otimes\beta_M^{}+ \beta_M^{}\otimes\id\big)\circ \dd + \dd\circ \big(\rho_M^{}\otimes\beta_M^{}+ \beta_M^{}\otimes\id\big)~.\label{eqn:tmpcausalhom2}
\end{flalign}
In this expression $\rho_M^{}\otimes\beta_M^{}$ and $\beta_M^{}\otimes\id$ are understood 
in terms of the tensor product of internal homomorphisms in $\dgVec$. Explicitly,
for $a\in\hoU\AA(M)^n$ and $a^\prime\in\hoU\AA(M)^{n^\prime}$, we have that
$\rho_M^{}\otimes\beta_M^{}(a\otimes a^\prime) = (-1)^{n}\, \rho_M^{}(a)\otimes \beta_M^{}(a^\prime)$
and $\beta_M^{}\otimes\id(a\otimes a^\prime) = \beta_M^{}(a)\otimes a^\prime$.
(These signs are crucial for verifying \eqref{eqn:tmpcausalhom2}.)
Inserting \eqref{eqn:tmpcausalhom2} into \eqref{eqn:tmpcausalhom} completes the proof.
\end{proof}

We next focus on the time-slice axiom. In the following we will always assume the
category fibered in groupoids $\pi:\Str\to\Loc$ to be Cauchy flabby, see Definition
\ref{def:flabby}. Let $f: M\to M^\prime$ be any Cauchy $\Loc$-morphism.
Due to Cauchy flabbiness, we may define an extension map 
\begin{subequations}
\begin{flalign}
\ext_f^{} : \pi^{-1}(M)_0 \longrightarrow \pi^{-1}(M^\prime)_0 ~,~~S\longmapsto \ext_f^{}S~,
\end{flalign}
such that for all $S\in \pi^{-1}(M)_0$ there exists a $\Str$-morphism
\begin{flalign}
f_{\sharp} : S\longrightarrow \ext_f^{}S~
\end{flalign}
\end{subequations}
with the property $\pi(f_{\sharp})=f:M\to M^\prime$.
\begin{lem}\label{lem:extensioncauchyflab}
Let $\pi:\Str\to\Loc$ be a Cauchy flabby category fibered in groupoids. Then:
\begin{itemize}
\item[(i)] For all Cauchy $\Loc$-morphisms $f:M\to M^\prime$ and 
all $S\in\pi^{-1}(M)_0$, there exists a unique $\pi^{-1}(M)$-morphism
$g_{(S,f)}^{} : S\to f^\ast\ext_f^{} S$ such that the diagram
\begin{flalign}
\xymatrix{
\ar[dr]_-{f_{\sharp}}S\ar[rr]^-{g_{(S,f)}^{}} && f^\ast\ext_f^{}S\ar[dl]^-{f_\ast}\\
&\ext_f^{} S&
}
\end{flalign}
in $\Str$ commutes.

\item[(ii)] For all Cauchy $\Loc$-morphisms $f:M\to M^\prime$ and 
all $S^\prime \in \pi^{-1}(M^\prime)_0$, there exists a (not necessarily unique) 
$\pi^{-1}(M^\prime)$-morphism $g^\prime_{(S^\prime,f)} : S^\prime\to \ext_f^{} f^\ast S^\prime$ 
such that the diagram
\begin{flalign}
\xymatrix{
S^\prime \ar[rr]^-{g^\prime_{(S^\prime,f)}} && \ext_f^{}f^\ast S^\prime\\
&\ar[ul]^-{f_\ast} f^\ast S^\prime \ar[ur]_-{f_\sharp}&
}
\end{flalign}
in $\Str$ commutes.
\end{itemize}
\end{lem}
\begin{proof}
Item (i) is a direct consequence of $f_\ast : f^\ast \ext_f^{} S \to \ext_f^{} S$ being cartesian,
cf.\ Definition \ref{def:cartesian}. Item (ii) follows from Cauchy flabbiness (cf.\ Definition \ref{def:flabby})
and the fact that $\pi(f_\ast) = \pi(f_\sharp) =f:M\to M^\prime$.
\end{proof}

Given any $\pi^{-1}(M)$-morphism $g : S\to\widetilde{S}$,
Cauchy flabbiness ensures existence of a
$\pi^{-1}(M^\prime)$-morphism $\ext_f^{} g : \ext_f^{} S \to \ext_f^{} \widetilde{S}$
such that the diagram
\begin{flalign}\label{eqn:extmorphism}
\xymatrix{
\ext_f^{} S\ar[rr]^-{\ext_f^{} g } && \ext_f^{} \widetilde{S}\\
\ar[u]^-{f_{\sharp}}S \ar[rr]_-{g} && \widetilde{S} \ar[u]_-{f_\sharp}
}
\end{flalign}
in $\Str$ commutes. (This follows from noticing that due to  $\pi(f_\sharp) = \pi(f_\sharp \circ g) =f:M\to M^\prime$
both  $f_\sharp : S\to\ext_f^{}S$ and $f_\sharp \circ g: S\to \ext_f^{}\widetilde{S}$ are extensions
of $S$ to $M^\prime$ along $f$.) However, the extended morphisms
$\ext_f^{} g$ in \eqref{eqn:extmorphism} in general will not be uniquely defined
by this diagram. For proving a homotopical generalization of the time-slice axiom,
we introduce a stronger version of the Cauchy flabbiness condition ensuring that the extended
morphisms are uniquely defined. (Examples are provided at the end of this section.)
\begin{defi}\label{def:strongCauchyflabby}
A category fibered in groupoids $\pi:\Str\to\Loc$ is called {\em strongly Cauchy flabby}
if it is Cauchy flabby (cf.\ Definition \ref{def:flabby}) and the $\pi^{-1}(M^\prime)$-morphisms
$g^\prime : S^\prime \to \widetilde{S^\prime}$ in \eqref{eqn:Cauchyflabbydiagram}
are uniquely specified by this diagram, for all Cauchy $\Loc$-morphisms
$f:M\to M^\prime$ and all $\Str$-morphisms $g: S\to S^\prime$
and $\widetilde{g}: S\to\widetilde{S^\prime}$ with the property
$\pi(g) = \pi(\widetilde{g}) = f:M\to M^\prime$.
\end{defi}
An immediate consequence of this definition is
\begin{cor}\label{cor:strongCauchy}
Let  $\pi:\Str\to\Loc$ be a strongly Cauchy flabby category fibered in groupoids.
Then the $\pi^{-1}(M^\prime)$-morphism $\ext_f^{} g : \ext_f^{} S \to \ext_f^{} \widetilde{S}$ 
is uniquely specified by \eqref{eqn:extmorphism}, for all Cauchy $\Loc$-morphisms
$f:M\to M^\prime$ and all $\pi^{-1}(M)$-morphisms $g : S\to\widetilde{S}$.
As a consequence, we obtain
\begin{flalign}
\ext_f^{} \id_S = \id_{\ext_f^{} S} \quad,\qquad \ext_f^{} g_1 \circ \ext_f^{} g_2 = \ext_f^{}(g_1\circ g_2)~,
\end{flalign}
for all objects $S$ in $\pi^{-1}(M)$ and all composable $2$-arrows
$S_0 \stackrel{g_1}{\longleftarrow} S_1 \stackrel{g_2}{\longleftarrow} S_2$
in $\pi^{-1}(M)$.
\end{cor}

We assume from now on the strong Cauchy flabbiness condition
of Definition \ref{def:strongCauchyflabby}. Given any Cauchy $\Loc$-morphism 
$f:M\to M^\prime$, Corollary \ref{cor:strongCauchy} allows us to define a $\dgAlg$-morphism
\begin{subequations}\label{eqn:extfast}
\begin{flalign}
\ext_f^\ast : \hoU\AA(M^\prime) \longrightarrow \hoU\AA(M)
\end{flalign}
going in the opposite direction of $\hoU\AA(f) : \hoU\AA(M)\to \hoU\AA(M^\prime)$ (cf.\ \eqref{eqn:hoUpiAAmorphexplicit}).
Explicitly, we set
\begin{flalign}
\big(\ext_f^\ast(a^\prime)\big)(S):= \AA(f_\sharp)^{-1}\big(a^\prime(\ext_f^{} S)\big)~,
\end{flalign}
for all $a^\prime\in\hoU\AA(M^\prime)^0 $ and $S\in\pi^{-1}(M)_0$, and
\begin{flalign}
\big(\ext_f^\ast(a^\prime)\big)(g_1,\dots,g_n) := \AA(f_\sharp)^{-1}\big(a^\prime(\ext_f^{} g_1,\dots,\ext_f^{} g_n)\big)~,
\end{flalign} 
\end{subequations}
for all $n\in \bbZ_{\geq 1}$, $a^\prime\in \hoU\AA(M^\prime)^n$ and $(g_1,\dots,g_n)\in\pi^{-1}(M)_n$.
(Here we also have used that $\AA : \Str\to\Alg$ satisfies the time-slice axiom in the sense of 
Definition \ref{def:QFT} in order to define the inverse $\AA(f_\sharp)^{-1}$.)
Using similar techniques as in \eqref{eqn:zetakappahomotopy}, we find that both 
compositions $\ext_f^{\ast}\circ\hoU\AA(f)$ and $\hoU\AA(f)\circ \ext_f^{\ast}$ are cochain homotopic
to the identity, i.e.\
\begin{subequations}\label{eqn:timeslicehomotopy}
\begin{flalign}
\ext_f^{\ast}\circ\hoU\AA(f) - \id_{\hoU\AA(M)} &= \dd\circ \phi_{f} + \phi_{f}\circ \dd~,\\
\hoU\AA(f)\circ \ext_f^{\ast} - \id_{\hoU\AA(M^\prime)} &= \dd\circ \bar \phi_{f} + \bar \phi_{f}\circ \dd~.
\end{flalign}
The cochain homotopies are obtained by using Lemma \ref{lem:extensioncauchyflab}. 
(Notice that for strongly Cauchy flabby $\pi:\Str \to \Loc$ 
the morphisms $g^\prime_{S^\prime,f}$ in Lemma \ref{lem:extensioncauchyflab} (ii) are unique.)
Explicitly, the first cochain homotopy is given by
\begin{flalign}
\big(\phi^{1}_f(a)\big)(S):= a(g_{(S,f)}^{-1})~,
\end{flalign}
for all $a\in \hoU\AA(M)^1$ and $S\in\pi^{-1}(M)_0$, and
\begin{flalign}
\big(\phi^{n+1}_f(a)\big)(g_1,\dots,g_n):= \sum_{i=0}^n (-1)^i~a\big(g_1,\dots, g_i, g_{(S_i,f)}^{-1}, (\ext_f^{}g_{i+1})^f,
\dots, (\ext_f^{} g_{n})^f\big)~,
\end{flalign}
for all $n\in\bbZ_{\geq 1}$, $a\in \hoU\AA(M)^{n+1}$ and $(g_1,\dots,g_n)\in\pi^{-1}(M)_n$. (The arrows
$(\ext_f^{}g)^f$ are defined as in \eqref{eqn:pullbackmorphallsame}.)
The second cochain homotopy explicitly reads as
\begin{flalign}
\big(\bar \phi_f^1(a^\prime)\big)(S^\prime):= a^\prime( g^{\prime\,-1}_{(S^\prime,f)})~,
\end{flalign}
for all $a^\prime\in \hoU\AA(M^\prime)^1$ and $S^\prime\in\pi^{-1}(M^\prime)_0$, and
\begin{flalign}
\big(\bar \phi_f^{n+1}(a^\prime)\big)(g_1^\prime,\dots,g_{n}^\prime) := \sum_{i=1}^n
(-1)^i~a^\prime\big(g_1^\prime,\dots,g_i^{\prime}, g_{(S_i^\prime,f)}^{\prime\,-1}, \ext_f^{}(g_{i+1}^{\prime\,f}),
\dots,\ext_f^{}(g_n^{\prime\,f})\big)~,
\end{flalign}
\end{subequations}
for all $n\in\bbZ_{\geq 1}$, 
$a^\prime\in \hoU\AA(M^\prime)^{n+1}$ and $(g_1^\prime,\dots,g_n^{\prime})\in\pi^{-1}(M^\prime)_n$.
This proves the following
\begin{theo}\label{theo:homotopicaltimeslice}
Let $\pi: \Str\to \Loc$ be a strongly Cauchy flabby category fibered in groupoids
and $\AA :\Str \to\Alg$ a quantum field theory in the sense of Definition \ref{def:QFT}.
Then the homotopy right Kan extension $\hoU\AA : \Loc\to \dgAlg$ (cf.\ \eqref{eqn:hoUpiAAfunctor}) 
satisfies the time-slice axiom `up to homotopy'. Explicitly, for any Cauchy $\Loc$-morphism
$f: M\to M^\prime$, the $\dgAlg$-morphism $\hoU\AA(f) : \hoU\AA(M)\to\AA(M^\prime)$
is inverted by the $\dgAlg$-morphism \eqref{eqn:extfast} up to the cochain homotopies in \eqref{eqn:timeslicehomotopy}.
\end{theo}
\begin{rem}
In analogy to Remarks \ref{rem:uptohomotopyfunctor} and \ref{rem:uptohomotopycommutativity},
a natural refinement of the time-slice axiom `up to homotopy' would be to promote also
the cochain homotopies \eqref{eqn:timeslicehomotopy} and their coherences to the data defining
a homotopical quantum field theory. Again, the systematic way to address this aspect is to use colored operads.
One of the main uses of the time-slice axiom in locally covariant quantum field theory
is to define the relative Cauchy evolution \cite{Brunetti,FewsterVerch}.
Notice that already our present non-operadic framework allows us to define a notion of relative Cauchy evolution 
for homotopical quantum field theories because we can invert up to homotopy 
all $\dgAlg$-morphisms $\hoU\AA(f) : \hoU\AA(M)\to\AA(M^\prime)$ corresponding to Cauchy $\Loc$-morphisms.
The homotopical relative Cauchy evolutions are then $\dgAlg$-endomorphism of $\hoU\AA(M)$ 
that are invertible `up to homotopy'. In particular, on the level of cohomologies
$H^\bullet(\hoU\AA(M))$ we obtain strict automorphism of graded algebras.
\end{rem}

We conclude this section by providing some examples of strongly 
Cauchy flabby categories fibered in groupoids.
\begin{ex}
Let $\pi : \Str\to \Loc$ be a Cauchy flabby category fibered in groupoids
such that for all objects $M$ in $\Loc$ the groupoid $\pi^{-1}(M)$ is discrete (i.e.\ the only morphisms are identities).
Then $\pi:\Str\to\Loc$ is also a strongly Cauchy flabby category fibered in groupoids.
An example of this situation is given by $\pi: \mathsf{LocSrc}_{D}\to\Loc$, see Example \ref{ex:source2}.
\end{ex}

\begin{ex}
Recall from Example \ref{ex:spin2} that the category fibered in groupoids
$\pi:\mathsf{SLoc}\to\Loc$ which describes spin structures is Cauchy flabby. 
It is also strongly Cauchy flabby: The extension $s^\prime$ of $s$
constructed in Example \ref{ex:spin2} is unique,
because both $s$ and $s^\prime$ are $\bbZ_2$-valued functions (hence locally constant)
and the image $f(M)$ is homotopic to $M^\prime$ for any Cauchy $\Loc$-morphism $f:M\to M^\prime$.
\end{ex}

\begin{ex}
Under the PDE-theoretic assumptions detailed in Example \ref{ex:gaugecon2},
the category fibered in groupoids 
$\pi: \mathsf{B}G^{\mathrm{con}}\Loc_{\mathrm{YM}}^{} \to\Loc$ is strongly Cauchy flabby.
The assumptions are in particular satisfied for $G=U(1)$ and any spacetime dimension $m$,
implying that $\pi: \mathsf{B}U(1)^{\mathrm{con}}\Loc_{\mathrm{YM}}^{} \to \Loc$
is strongly Cauchy flabby.
\end{ex}

\begin{ex}
The category fibered in groupoids $\pi: \Loc \times G \to\Loc$ (cf.\ Examples \ref{ex:globalgauge} and \ref{ex:globalgauge2})
corresponding to a ``global gauge group'' $G$  is strongly Cauchy flabby.
\end{ex}

\section*{Acknowledgments}
We would like to thank Ulrich Bunke and Urs Schreiber for useful comments.
We also would like to thank the anonymous referees for
their comments and suggestions which helped us to clarify and improve parts of
the paper. The work of M.B.\ is supported by a Postdoctoral Fellowship 
of the Alexander von Humboldt Foundation (Germany). 
A.S.\ gratefully acknowledges the financial support of 
the Royal Society (UK) through a Royal Society University 
Research Fellowship.

\appendix

\section{\label{app:dgAlg}Differential graded vector spaces and algebras}
\subsection{Basics}
We briefly recall some basic notions of differential graded vector spaces (i.e.\ cochain complexes of vector spaces)
and differential graded algebras. We denote the underlying field by $\bbK$. 

\begin{defi}
A differential graded vector space (in non-negative degrees) is a pair $(V^\bullet,\dd_V^{})$, where
$V^\bullet = \{V^n\}_{n\in\bbZ_{\geq 0}}$ is a family of vector spaces 
and $\dd_{V}^{} = \{\dd_V^n : V^n \to V^{n+1}\}_{n\in\bbZ_{\geq 0}}$ 
is a family of linear maps satisfying $\dd_V^{n+1} \circ \dd_V^{n} =0$, for all $n\in\bbZ_{\geq 0}$.
A morphism $L : (V^\bullet,\dd_{V}^{}) \to (W^\bullet,\dd_{W}^{})$ of differential graded vector spaces
is a family of linear maps $L = \{L^n : V^n \to W^n\}_{n\in\bbZ_{\geq 0}}$ satisfying
$\dd_{W}^{n}\circ L^n= L^{n+1}\circ \dd_V^{n}$, for all $n\in\bbZ_{\geq 0}$.
We denote by $\dgVec$ the category of differential graded vector spaces (in non-negative degrees).
\end{defi}
\begin{rem}
In order to simplify notations, we shall denote objects in $\dgVec$ simply by symbols
like $V^\bullet$ suppressing the differentials $\dd_V^{} : V^\bullet \to V^{\bullet+1}$ from the notation. 
Moreover, we shall denote all differentials simply by $\dd$ (without subscript and superscript)
as it will be clear from the context on which graded vector space and degree they act. 
\end{rem}

Recall that the category $\dgVec$ is monoidal: The tensor product of two objects 
$V^\bullet$ and $W^\bullet$ in $\dgVec$ is given by 
\begin{subequations}
\begin{flalign}
V^\bullet \otimes W^\bullet := \Big\{\bigoplus_{i+j=n} V^i\otimes W^j\Big\}_{n\in\bbZ_{\geq 0}}~,
\end{flalign}
together with the differential specified by
\begin{flalign}
\dd (v \otimes w) = \dd v \otimes w + (-1)^{i} \, v\otimes\dd w~,
\end{flalign}
\end{subequations}
for all $v\in V^i$ and $w\in W^j$. The monoidal unit is the object $I^\bullet$
in $\dgVec$ with $I^0=\bbK$, $I^n=0$, for all $n\geq 1$, and trivial differential $\dd =0$. 

\begin{defi}\label{def:dgAlg}
A differential graded algebra (in non-negative degrees) is a monoid object in $\dgVec$. More explicitly,
it is an object $A^\bullet$ in $\dgVec$ together with two $\dgVec$-morphisms
$\mu_{A^\bullet}^{} : A^\bullet\otimes A^\bullet \to A^\bullet$ (called product)
and $\eta_{A^\bullet}^{} : I^\bullet \to A^\bullet$ (called unit), such that the diagrams
\begin{subequations}
\begin{flalign}
\xymatrix{
\ar[d]_-{\id_{A^\bullet}\otimes\mu_{A^\bullet}^{}} A^\bullet\otimes A^\bullet\otimes A^\bullet \ar[rr]^-{\mu_{A^\bullet}^{}\otimes \id_{A^\bullet}}&& A^\bullet\otimes A^\bullet \ar[d]^-{\mu_{A^\bullet}^{}}\\
A^\bullet\otimes A^\bullet  \ar[rr]_-{\mu_{A^\bullet}^{}}&& A^\bullet
}
\end{flalign}
\begin{flalign}
\xymatrix{
\ar[drr]_-{\simeq} I^\bullet \otimes A^\bullet \ar[rr]^-{\eta_{A^\bullet}^{}\otimes\id_{A^\bullet}} && \ar[d]_-{\mu_{A^\bullet}^{}}A^\bullet\otimes A^\bullet && A^\bullet\otimes I^\bullet\ar[ll]_-{\id_{A^\bullet}\otimes\eta_{A^\bullet}^{}}\ar[dll]^-{\simeq}\\
&& A^\bullet &&
}
\end{flalign}
\end{subequations}
in $\dgVec$ commute. A morphism
$\kappa : (A^\bullet,\mu_{A^\bullet}^{},\eta_{A^\bullet}^{})\to (B^\bullet,\mu_{B^\bullet}^{},\eta_{B^\bullet}^{})$ 
of differential graded algebras
is a $\dgVec$-morphism $\kappa : A^\bullet\to B^\bullet$ that preserves products and units, i.e.\
$\mu_{B^\bullet}^{} \circ(\kappa\otimes\kappa) = \kappa\circ \mu_{A^\bullet}^{}$ and 
$\eta_{B^\bullet}^{} = \kappa \circ \eta_{A^\bullet}^{}$.
We denote by $\dgAlg$ the category of differential graded algebras (in non-negative degrees).
\end{defi}
\begin{rem}
Notice that our differential graded algebras are not assumed to be (graded) commutative. 
In order to simplify notations, we shall denote objects in $\dgAlg$ simply by symbols
like $A^\bullet$ suppressing the product $\mu_{A^\bullet}^{}$
and unit $\eta_{A^\bullet}^{}$ from the notation. We shall often use juxtaposition 
$a\,a^\prime := \mu_{A^\bullet}^{} (a\otimes a^\prime) $
to denote products and the symbol $\1 := \eta_{A^\bullet}^{}(1) \in A^0$ to denote the unit element.
By definition, we have the following properties
\begin{flalign}
\dd (a\,a^\prime) = (\dd a)\, a^\prime + (-1)^{i}\,a\,(\dd a^\prime) \quad,\qquad \dd\1=0~,
\end{flalign}
for all $a\in A^i$ and $a^\prime\in A^j$.
\end{rem}

We obviously have a forgetful functor
\begin{subequations}
\begin{flalign}
\mathrm{Forget} : \dgAlg\longrightarrow \dgVec
\end{flalign}
that assigns to an object $A^\bullet$ in $\dgAlg$ its underlying
differential graded vector space, i.e.\ forgets the product $\mu_{A^\bullet}^{}$ 
and unit $\eta_{A^\bullet}^{}$. The forgetful functor has a left adjoint
\begin{flalign}
\mathrm{Free} : \dgVec\longrightarrow \dgAlg
\end{flalign}
\end{subequations}
given by the free $\dgAlg$-construction. Explicitly, given any object 
$V^\bullet$ in $\dgVec$,  we have
\begin{flalign}\label{eqn:freedgAlg}
\mathrm{Free}(V^\bullet) := \bigoplus_{k=0}^{\infty} {V^\bullet}^{\otimes k}~,
\end{flalign}
where $\bigoplus$ denotes the coproduct in $\dgVec$ and by convention ${V^\bullet}^{\otimes 0} = I^\bullet$.
The product $\mu_{\mathrm{Free}(V^\bullet)}^{}$ is simply given by the identification
${V^\bullet}^{\otimes k}\otimes {V^\bullet}^{\otimes l}\simeq {V^\bullet}^{\otimes(k+l)}$
and the unit $\eta_{\mathrm{Free}(V^\bullet)}^{}$ is given by mapping $I^\bullet$ via the 
identity to the component corresponding to $k=0$ in \eqref{eqn:freedgAlg}.
From this explicit description, it is easy to show that we have constructed an adjunction
\begin{flalign}\label{eqn:adjunction}
\xymatrix{
\mathrm{Free} : \dgVec ~\ar@<0.5ex>[r]&\ar@<0.5ex>[l]  ~\dgAlg : \mathrm{Forget}
}
\end{flalign}
with $\mathrm{Forget}$ being the right adjoint.

\subsection{Model category structures}
Both of our categories $\dgVec$ and $\dgAlg$ can be equipped with model category structures,
see e.g.\ \cite{Dwyer} for a concise introduction to model categories.
\begin{theo}[\cite{Dwyer}]\label{theo:dgVecmodel}
Define a morphism $L : V^\bullet\to W^\bullet$ in $\dgVec$ to be
\begin{itemize}
\item[(i)] a weak equivalence if $L$ induces a cohomology isomorphism
$H^\bullet (L) : H^\bullet(V^\bullet) \to H^\bullet(W^\bullet)$;
\item[(ii)] a fibration if $L^n : V^n\to W^n$ is surjective, for all $n\in\bbZ_{\geq 0}$;
\item[(iii)] a cofibration if $L$ has the left lifting property with respect to all morphisms
which are fibrations and weak equivalences (i.e.\ all acyclic fibrations).
\end{itemize}
Then with these choices $\dgVec$ is a model category.
\end{theo}
\begin{theo}[\cite{Jardine}]\label{theo:dgAlgmodel}
Define a morphism $\kappa : A^\bullet\to B^\bullet$ in $\dgAlg$ to be
\begin{itemize}
\item[(i)] a weak equivalence if $\kappa$ induces a cohomology isomorphism
$H^\bullet (\kappa) : H^\bullet(A^\bullet) \to H^\bullet(B^\bullet)$;
\item[(ii)] a fibration if $\kappa^n : A^n\to B^n$ is surjective, for all $n\in\bbZ_{\geq 0}$;
\item[(iii)] a cofibration if $\kappa$ has the left lifting property with respect to all morphisms
which are fibrations and weak equivalences (i.e.\ all acyclic fibrations).
\end{itemize}
Then with these choices $\dgAlg$ is a model category.
\end{theo}
\begin{rem}
Notice that in $\dgVec$ all objects are fibrant, i.e.\
the unique morphism $V^\bullet \to 0$ from any object $V^\bullet$ to the terminal object
is a fibration. The same holds true for $\dgAlg$.
\end{rem}

From the definition of the model structures on $\dgVec$ and $\dgAlg$
we immediately observe
\begin{propo}\label{prop:Quillenadjunction}
The adjunction \eqref{eqn:adjunction} is a Quillen adjunction, i.e.\ the right adjoint
functor $\mathrm{Forget}: \dgAlg\to\dgVec$ preserves fibrations and acyclic fibrations.
Moreover, $\mathrm{Forget}: \dgAlg\to\dgVec$ preserves weak equivalences
and it even detects them, i.e., given a $\dgAlg$-morphism $\kappa$, 
if $\mathrm{Forget}(\kappa)$ is a weak equivalence in $\dgVec$ 
then $\kappa$ is a weak equivalence in $\dgAlg$.
\end{propo}

\subsection{Homotopy limits in $\dgVec$}
We recall how homotopy limits may be computed in the model category $\dgVec$.
For more details we refer to \cite[Section 16.8]{Dugger} and \cite[Appendix B]{BSShomotopy};
see also \cite{Dwyer,Hirschhorn} for an introduction to the abstract theory of homotopy (co)limits.
\sk

Let $\DD$ be a small category. Recall that the nerve of $\DD$ is the simplicial set $\{\DD_n\}_{n\in\bbN_0}$,
where $\DD_0$ is the set of objects in $\DD$ and $\DD_n$, for $n\geq 1$,  is the set
of all composable $n$-arrows in $\DD$. For $n\geq 1$, we shall denote an element of $\DD_n$
by an $n$-tuple $(f_1,\dots,f_n)$ of morphisms in $\DD$ such that
the source of $f_i$ is the target of $f_{i+1}$ (i.e.\ the compositions $f_i\circ f_{i+1}$ exist).
The face maps are given by composing two subsequent arrows (or throwing away the first/last arrow)
and the degeneracy maps are given by inserting the identity morphisms.
\sk

Given a functor $X^\bullet : \DD \to \dgVec$, which we shall interpret as a diagram of shape $\DD$ in
$\dgVec$, the homotopy limit $\holim_{\dgVec}^{} X^\bullet$ is an object in $\dgVec$ 
that may be computed by the following three-step construction,
cf.\ \cite[Section 16.8]{Dugger} and \cite[Appendix B]{BSShomotopy}:
First, we take the cosimplicial replacement of $X^\bullet :\DD \to \dgVec$ resulting in a cosimplicial object in $\dgVec$.
Second, we assign to this cosimplicial object a double cochain complex of vector spaces
via the co-normalized Moore complex construction. 
Third, we define $\holim_{\dgVec}^{} X^\bullet$ 
to be the $\prod$-total complex of this double complex.
The result of this construction is rather explicit and reads as follows:
The double cochain complex $X^{\bullet,\bullet}$ has components
\begin{flalign}\label{eqn:doublecomplex}
X^{0,\bullet} = \prod_{d\in \DD_0} X(d)^\bullet \quad,\qquad X^{n,\bullet}= \prod_{\mycom{(f_1,\dots,f_n)\in \DD_n}{f_i\neq \id}} X(\mathrm{t}(f_1))^\bullet~,
\end{flalign}
for $n\in \bbZ_{\geq 1}$, where $\mathrm{t}(f)$ denotes the target of the $\DD$-morphism $f$.
It is very convenient to regard elements $x\in X^{n,m}$ as mappings
\begin{subequations}
\begin{flalign}
\DD_n\ni (f_1,\dots,f_n) \longmapsto x(f_1,\dots,f_n)\in X(\mathrm{t}(f_1))^m
\end{flalign}
on all of $\DD_n$,  which satisfy the normalization conditions
\begin{flalign}
x(f_1,\dots,f_{i-1}, \id,f_{i+1}, \dots f_n) =0~,
\end{flalign}
\end{subequations}
for all $i=1,\dots,n$.
The vertical differential $\dd^\mathrm{v} : X^{\bullet,\bullet}\to X^{\bullet +1 ,\bullet}$ in this notation
reads as
\begin{subequations}
\begin{multline}
 \dd^{\mathrm{v}}x(f_1,\dots,f_{n+1}) = X(f_1)\big(x(f_2,\dots,f_{n+1})\big) \\
 + \sum_{i=1}^n (-1)^i\, x(f_1,\dots,f_i\circ f_{i+1},\dots, f_{n+1}) 
+ (-1)^{n+1} \,x(f_1,\dots,f_n)~,
\end{multline}
for all $x\in X^{n,m}$, and the horizontal differential $\dd^\mathrm{h} : X^{\bullet,\bullet}\to X^{\bullet ,\bullet+1}$ 
is simply given by
\begin{flalign}
\dd^{\mathrm{h}}x(f_1,\dots,f_{n}) = \dd \big(x(f_1,\dots,f_n)\big)~,
\end{flalign}
\end{subequations}
for all $x\in X^{n,m}$, where $\dd$ on the right-hand side is the differential on $X(\mathrm{t}(f_1))^\bullet$.
The homotopy limit $\holim_{\dgVec}^{} X^\bullet$  is then the graded vector space with components
\begin{subequations}\label{eqn:holimdgVec}
\begin{flalign}
(\holim_{\dgVec}X^\bullet)^k = \prod_{n+m=k} X^{n,m}~,
\end{flalign}
for all $k\in\bbZ_{\geq 0}$, and differential given by
\begin{flalign}
\dd^{\mathrm{tot}} = \dd^\mathrm{v} + (-1)^n\, \dd^{\mathrm{h}}
\end{flalign}
\end{subequations}
on the factor $X^{n,m}$.
\sk

It is easy to see that the assignment of the object $\holim_{\dgVec}^{} X^\bullet$ in $\dgVec$
to a diagram $X^\bullet :\DD\to\dgVec$ of shape $\DD$ is functorial, hence
we obtain a homotopy limit functor
\begin{flalign}\label{eqn:holimfunctordgVec}
\holim_{\dgVec} : \dgVec^\DD \longrightarrow \dgVec~,
\end{flalign}
where $\dgVec^\DD$ is the category of functors from $\DD$ to $\dgVec$.

\subsection{Homotopy limits in $\dgAlg$}
Let again $\DD$ be a small category. The goal of this subsection is to describe the
homotopy limit functor for the model category $\dgAlg$.
\sk

Given any diagram $X^\bullet : \DD\to \dgAlg$, let us forget 
for the moment the $\dgAlg$-structure and form the homotopy limit
\eqref{eqn:holimdgVec} in $\dgVec$. On the resulting object
$\holim_{\dgVec}X^\bullet$ in $\dgVec$ we may define a  product and unit
by setting
\begin{subequations}\label{eqn:holimdgAlg}
\begin{flalign}
(x\,x^\prime)(f_1,\dots, f_{n+n^\prime}) := (-1)^{m\, n^\prime} x(f_1,\dots,f_n) ~ X(f_1\circ \cdots\circ f_n)\big(x^\prime(f_{n+1},\dots,f_{n+n^\prime})\big)~,
\end{flalign}
for all $x\in X^{n,m}$ and $x^\prime\in X^{n^\prime,m^\prime}$,
and
\begin{flalign}
\1(d) := \1 \in X(d)^{0}~.
\end{flalign}
\end{subequations}
It is easy to check that the product is associative and compatible with the unit (cf.\ Definition \ref{def:dgAlg}).
Moreover, a slightly lengthy computation shows that
product and unit are also compatible with the differential of $\holim_{\dgVec}X^\bullet$ in the sense that
\begin{flalign}
\dd^{\mathrm{tot}}(x\,x^\prime) = (\dd^{\mathrm{tot}} x) \, x^\prime + (-1)^{n+m}\, x\,(\dd^{\mathrm{tot}} x^\prime)\quad,
\qquad \dd^{\mathrm{tot}}\1 =0~,
\end{flalign}
for all $x\in X^{n,m}$ and $x^\prime\in X^{n^\prime,m^\prime}$. 
As a consequence, we may equip for any diagram $X^\bullet : \DD\to \dgAlg$
the differential graded vector space $\holim_{\dgVec}^{}X^\bullet$ with the structure of a differential graded algebra.
This $\dgAlg$-structure is natural in the sense that given any morphism
$\kappa : X^\bullet\to Y^\bullet$ in $\dgAlg^\DD$, the $\dgVec$-morphism
$\holim_{\dgVec}^{}\kappa : \holim_{\dgVec}^{}X^\bullet\to \holim_{\dgVec}^{}Y^\bullet$
preserves products and units, hence it is a $\dgAlg$-morphism. We thus have obtained a
functor from $\dgAlg^\DD$ to $\dgAlg$ which we shall denote by
\begin{flalign}\label{eqn:holimfunctordgAlg}
\holim_{\dgAlg} : \dgAlg^\DD \longrightarrow \dgAlg~.
\end{flalign}
Notice that due to the `same' choice of weak equivalences in 
$\dgVec$ and $\dgAlg$ (cf.\ Theorems \ref{theo:dgVecmodel} and \ref{theo:dgAlgmodel}),
$\holim_{\dgAlg}^{}$ is clearly a homotopy functor (i.e.\ it preserves weak equivalences)
as $\holim_{\dgVec}^{}$ is a homotopy functor.
\sk

It remains to show that \eqref{eqn:holimfunctordgAlg} is a homotopy limit functor for $\dgAlg$.
Using \cite[Theorem 2.3.7]{Walter} and our Quillen adjunction from Proposition \ref{prop:Quillenadjunction},
this will be the case provided that we can verify the following properties, for
all diagrams $X^\bullet : \DD\to\dgAlg$:
\begin{enumerate}
\item $\mathrm{Forget}\big(\holim_{\dgAlg}^{} X^\bullet\big) = \holim_{\dgVec}^{}\mathrm{Forget}^{\DD}(X^\bullet)$;
\item there exists a natural $\dgAlg$-morphism $e_{X^\bullet}^{} : \lim X^\bullet \to \holim_{\dgAlg}^{}X^\bullet$, where
$\lim$ denotes the ordinary categorical limit in $\dgAlg$;
\item $\mathrm{Forget}(e_{X^\bullet}^{})$ is the canonical $\dgVec$-morphism
$\lim \mathrm{Forget}^{\DD}(X^\bullet) \to \holim_{\dgVec}^{} \mathrm{Forget}^{\DD}(X^\bullet)$.
\end{enumerate}
\sk

Notice that item 1.\ holds true on account of our definition of $\holim_{\dgAlg}^{}$. 
For items 2.\ and 3.\ we have to form $\lim X^\bullet$ in $\dgAlg$ 
as well as $\lim \mathrm{Forget}^{\DD}(X^\bullet)$ in $\dgVec$. 
Recalling that limits in $\dgAlg$ may be computed by first computing the limit in $\dgVec$
and then equipping the result with an appropriate product and unit,
we start with the second task and form the limit $\lim \mathrm{Forget}^{\DD}(X^\bullet)$ in $\dgVec$.
Its underlying graded vector space has components
\begin{subequations}\label{eqn:limdgVec}
\begin{flalign}
\big(\lim \mathrm{Forget}^{\DD}(X^\bullet)\big)^k = \Big\{x \in \prod_{d\in \DD_0} X(d)^k\,:\; X(f)\big(x(\mathrm{s}(f))\big) = x(\mathrm{t}(f))\,,~\forall f\in \DD_1\Big\}~,
\end{flalign}
for all $k\in\bbZ_{\geq0}$, where $\mathrm{s}(f)$ denotes the source of the $\DD$-morphism $f$.
The differential $\dd^{\lim}$ on $\lim \mathrm{Forget}^{\DD}(X^\bullet)$ is given by 
\begin{flalign}
\dd^{\lim}x(d)= \dd \big(x(d)\big)~,
\end{flalign}
\end{subequations}
for all $x\in (\lim \mathrm{Forget}^{\DD}(X^\bullet))^k$, where $\dd$ on the right-hand side 
is the differential on $X(d)^\bullet$.
\sk

From \eqref{eqn:holimdgVec}, \eqref{eqn:doublecomplex} and \eqref{eqn:limdgVec}, 
we observe that there exists a natural $\dgVec$-morphism 
\begin{subequations}\label{eqn:eForget}
\begin{flalign}
e_{\mathrm{Forget}^{\DD}(X^\bullet)} :  \lim \mathrm{Forget}^{\DD}(X^\bullet) \longrightarrow \holim_{\dgVec}\mathrm{Forget}^{\DD}(X^\bullet)~.
\end{flalign}
Explicitly, the $k$-th component
\begin{flalign}
e_{\mathrm{Forget}^\DD(X^\bullet)} : \big( \lim \mathrm{Forget}^{\DD}(X^\bullet)\big)^k\subseteq
X^{0,k} \longrightarrow \big(\holim_{\dgVec}\mathrm{Forget}^{\DD}(X^\bullet)\big)^{k}=\prod_{n+m=k}X^{n,m} ~
\end{flalign}
is induced by the canonical inclusion 
\begin{flalign}
X^{0,k} \hookrightarrow \prod_{n+m=k}X^{n,m}
\end{flalign}
\end{subequations}
in the cartesian product.
\sk

In order to obtain the limit $\lim X^\bullet$ in $\dgAlg$, 
we endow the differential graded vector space $\lim \mathrm{Forget}^{\DD}(X^\bullet)$ 
given in \eqref{eqn:limdgVec} with a suitable product and unit. 
Given $x \in (\lim \mathrm{Forget}^{\DD}(X^\bullet))^k$ and 
$x^\prime \in (\lim \mathrm{Forget}^{\DD}(X^\bullet))^{k^\prime}$, we set 
\begin{subequations}\label{eqn:limdgAlg}
\begin{flalign}
(x\,x^\prime) (d) := x(d)\,x^\prime(d)~.
\end{flalign}
It is straightforward to check that this product is associative and compatible with the differential
$\dd^{\lim}$. The unit element $\1$ is defined as in \eqref{eqn:holimdgAlg}, i.e.\
\begin{flalign}
\1(d) := \1 \in \big(\lim \mathrm{Forget}^{\DD}(X^\bullet)\big)^0~,
\end{flalign} 
\end{subequations}
and it is clear that $\dd^{\lim}\1=0$.
We shall denote the resulting differential graded algebra by
$\lim X^\bullet$ and note that it fulfills the universal property 
for the limit of $X^\bullet: \DD \to \dgAlg$. 
\sk 

Note that $e_{\mathrm{Forget}^{\DD}(X^\bullet)}$ given in \eqref{eqn:eForget} 
is compatible with the products and units we introduced on the source \eqref{eqn:limdgAlg} 
and on the target \eqref{eqn:holimdgAlg}.
Hence, the $\dgVec$-morphism $e_{\mathrm{Forget}^{\DD}(X^\bullet)}$ defines 
a $\dgAlg$-morphism 
\begin{flalign}
e_{X^\bullet}^{}: \lim X^\bullet \longrightarrow \holim_{\dgAlg} X^\bullet~.
\end{flalign}
These $e_{X^\bullet}^{}$ are natural with respect to morphisms in the functor category $\dgAlg^\DD$, 
thus showing that the requirement of item 2.\ is fulfilled. 
By construction, $\mathrm{Forget}(e_{X^\bullet}^{})$ coincides 
with the canonical $\dgVec$-morphism $e_{\mathrm{Forget}^{\DD}(X^\bullet)}: \lim \mathrm{Forget}^{\DD}(X^\bullet) 
\to \holim_{\dgVec}^{} \mathrm{Forget}^{\DD}(X^\bullet)$, as required by item 3. 
Therefore, by \cite[Theorem 2.3.7]{Walter}, together with the Quillen adjunction 
of Proposition \ref{prop:Quillenadjunction}, the proof of the following statement is complete. 
\begin{cor}
\eqref{eqn:holimfunctordgAlg} is a homotopy limit functor for $\dgAlg$.
\end{cor}

\end{document}